\documentclass[
submission
]{dmtcs-episciences}

\usepackage[utf8]{inputenc}
\usepackage[T1]{fontenc}
\usepackage[round]{natbib}
\usepackage{amssymb}
\usepackage{amsthm}
\usepackage{preamble}
\usepackage[capitalise,noabbrev]{cleveref}

\graphicspath{{figures/}}

\theoremstyle{plain}
\newtheorem{theorem}{Theorem}
\newtheorem{lemma}[theorem]{Lemma}
\newtheorem{claim}[theorem]{Claim}
\newtheorem{corollary}[theorem]{Corollary}
\newtheorem{proposition}[theorem]{Proposition}

\theoremstyle{definition}
\newtheorem{definition}[theorem]{Definition}

\theoremstyle{remark}

\newtheorem{observation}[theorem]{Observation}
\newtheorem*{claim*}{Claim}
\newtheorem*{note*}{Note}

\newenvironment{claimproof}{%
  \begin{proof}[of the Claim]%
}{%
  \end{proof}%
}
\newenvironment{romanenumerate}{\begin{enumerate}[(i)]}{\end{enumerate}}
\let\proofsubparagraph\paragraph

\title{Parameterized Spanning Tree Congestion\thanks{An extended abstract of this work was presented at the
50th International Symposium on Mathematical Foundations of Computer Science (MFCS 2025)~\citep{mfcs/LampisMNOV025}.}}

\author{
Michael Lampis\affiliationmark{1}\thanks{Supported by ANR project ANR-21-CE48-0022 (S-EX-AP-PE-AL).}
\and Valia Mitsou\affiliationmark{2}
\and Edouard Nemery\affiliationmark{1}\\
\and Yota Otachi\affiliationmark{3}\thanks{Supported by JSPS KAKENHI Grant Numbers JP21K11752, JP22H00513, and JP24H00697.}
\and Manolis Vasilakis\affiliationmark{1}
\and Daniel Vaz\affiliationmark{4}}

\affiliation{
Universit\'{e} Paris-Dauphine, PSL University, CNRS, LAMSADE, France\\
Universit\'{e} Paris Cit\'{e}, CNRS, IRIF, France\\
Nagoya University, Japan\\
LIGM, Universit\'{e} Gustave Eiffel, CNRS, ESIEE Paris, France}

\keywords{Parameterized complexity, Treewidth, Graph width parameters}

\begin{document}


\maketitle

\begin{abstract}
In this paper we study the \textsc{Spanning Tree Congestion} problem, where we are given an undirected graph $G=(V,E)$
and are asked to find a spanning tree $T$ of minimum \emph{congestion}.
Here, the congestion of a tree $T$ is the maximum of the edge congestion over $e\in T$, where the edge congestion of $e$ is the number of edges $uv\in E$ such
that the (unique) path from $u$ to $v$ in $T$ traverses $e$. We consider this
well-studied NP-hard problem from the point of view of (structural)
parameterized complexity and obtain the following results:

\begin{itemize}

\item We resolve a natural open problem by showing that \textsc{Spanning Tree Congestion} is not FPT
parameterized by treewidth (under standard assumptions). More strongly, we
present a generic reduction which applies to (almost) any parameter of the form
``vertex-deletion distance to class $\mathcal{C}$'', thus obtaining
W[1]-hardness for more restricted parameters, including
tree-depth plus feedback vertex number, or incomparable to treewidth, such as twin
cover.  Via a slight tweak of the same reduction we also show that the problem
is NP-complete on interval graphs of modular-width $4$.

\item Even though it is known that \textsc{Spanning Tree Congestion} remains NP-hard on instances with only
one vertex of unbounded degree, it is currently open whether the problem
remains hard on bounded-degree graphs. We resolve this question by showing
NP-hardness on graphs of maximum degree $8$.

\item Complementing the problem's W[1]-hardness for treewidth, we formulate an
algorithm that runs in time roughly ${(k+w)}^{\mathcal{O}(w)}$, where $k$ is the desired
congestion and $w$ the treewidth, improving a previous argument for parameter
$k+w$ that was based on Courcelle's theorem.  This explicit algorithm pays off
in two ways: it allows us to obtain an FPT approximation scheme for parameter
treewidth, that is, a $(1+\varepsilon)$-approximation running in time roughly
${(w/\varepsilon)}^{\mathcal{O}(w)}$; and it leads to an exact FPT algorithm for parameter
clique-width$+k$ via a Win/Win argument.

\item Finally, motivated by the problem's hardness for most standard structural
parameters, we present FPT algorithms for several more restricted cases,
namely, for the parameters vertex-deletion distance to clique; vertex
integrity; and feedback edge number, in the latter case also achieving a
single-exponential running time dependence on the parameter.

\end{itemize}

\end{abstract}

\clearpage

\section{Introduction}\label{sec:introduction}

One of the most well-studied types of problems in network optimization involves
finding, for a given graph $G$, a spanning tree of $G$ that optimizes a certain
objective. In this paper we focus on a well-known problem of this type called
\STC.  The motivation of this problem can be summarized as follows: Every edge
$e$ of a spanning tree $T$ is selected with the goal of maintaining
connectivity between the two parts of the graph given by the two components of
$T-e$.  We can then think of every other edge $e'$ with endpoints in both
components of $T-e$ as being ``simulated'' by a path in $T$ that traverses $e$;
hence, the more such edges exist, the more $e$ is used and ``congested''.  Our
optimization goal, then, is to find a tree where all edges have congestion as
low as possible, because in such a tree each selected edge is responsible for
simulating only a small number of non-selected edges and therefore the tree can
be thought of as a sparse approximate representation of the original graph.
Equivalently, for a spanning tree $T$ of $G$, we say that the \emph{detour} of
an edge $\{u, v\} \in E(G)$ in $T$ is the unique $u$--$v$ path in~$T$.  The
number of detours that traverse an edge in $T$ constitutes its congestion,
while the congestion of $T$ is defined as the maximum over the congestion of
all of its edges.%
\footnote{This has also been referred to as the \emph{edge remember number} of $G$
relative to $T$ in the literature~\citep[Section~11]{tcs/Bodlaender98}.}
The \emph{spanning tree congestion} of $G$, denoted by $\stc(G)$, is the minimum congestion over all of its spanning trees.
In this paper we focus on the decision version of the problem, which asks, given $G$ and an integer $k$, whether $\stc(G) \le k$;
we refer to this problem as \STC, as we do not consider the optimization variant in our work.

Spanning trees of low congestion are a natural notion that is well-studied both
from the combinatorial and the algorithmic point of view. Unsurprisingly, {\STC}
is NP-complete~\citep[Section~5.6]{phd_thesis}. It therefore makes sense to
study the parameterized complexity of this problem, as parameterized complexity
is one of the main tools for dealing with computational intractability.\footnote{Throughout the paper we assume that the reader is familiar with the
basics of parameterized complexity, as given in standard
textbooks~\citep{books/CyganFKLMPPS15}.} The most natural parameter one could
consider is perhaps the objective value $k$, but unfortunately the problem is
known to be NP-hard for all fixed $k\ge
5$~\citep{algorithmica/BodlaenderFGOL12,algorithmica/LuuC24}. This motivates us to
focus on \emph{structural} graph parameters, where much less is currently
known. Indeed, it is so far open whether {\STC} is fixed-parameter tractable for
treewidth, which is the most widely studied parameter of this type (this is
mentioned as an open problem by~\cite{birthday/Otachi20}).  What is known,
however, is that the problem is FPT when parameterized by both treewidth and
$k$~\citep{algorithmica/BodlaenderFGOL12} and that the problem is NP-hard on
graphs of clique-width at most $3$ (implied by the NP-hardness on chain
graphs~\citep{jgaa/OkamotoOUU11}).

\begin{figure}[ht]
  \centering
      \begin{tikzpicture}[every node/.style={thick, align=center}, scale=0.9]
        \small

        \node[fpt] (vc) at (0,0) {$\vc$};
        \node[fpt,above = 1 of vc] (vi) {$\vi$};
        \node[w1,above = 2 of vc] (td) {$\td$};
        \node[w1,above = 2 of vi] (tw) {$\tw$};
        \node[nph,above = 1 of tw] (cw) {$\cw$};
        \draw[thick,<-] (cw) -- (tw);
        \draw[thick,<-] (tw) -- (td);
        \draw[thick,<-] (td) -- (vi);
        \draw[thick,<-] (vi) -- (vc);

        \node[fpt,left = 1 of vc] (fes) {$\fes$};
        \node[w1,above = 1 of fes] (fvs) {$\fvs$};
        \draw[thick,<-] (fvs) -- (fes);
        \draw[thick,<-] (tw) edge [bend right] (fvs);
        \draw[thick,<-] (fvs)--(vc);

        \node[open,right = 0.5 of vi] (nd) {$\nd$};
        \node[nph,above = 2 of nd] (mw) {$\mw$};
        \draw[thick,<-] (nd) edge [bend left] (vc);
        \draw[thick,<-] (mw) -- (nd);
        \draw[thick,<-] (cw) edge [bend left] (mw);

        \node[w1,right = 0.5 of nd] (tc) {$\tc$};
        \node[nph,above = 2 of tc] (srb) {$\srb$};
        \draw[thick,<-] (tc) edge [bend left] (vc);
        \draw[thick,<-] (mw) -- (tc);
        \draw[thick,<-] (cw) edge [bend left] (srb);
        \draw[thick,<-] (srb) edge [bend left] (td);
        \draw[thick,<-] (srb) edge [bend left] (nd);

        \node[fpt,right = 2.5 of vc] (dtc) {$\dtc$};
        \node[w1,above = 2 of dtc] (cvd) {$\cvd$};
        \draw[thick,<-] (srb) -- (cvd);
        \draw[thick,<-] (nd) -- (dtc);
        \draw[thick,<-] (cvd) -- (dtc);
        \draw[thick,<-] (cvd) -- (tc);
      \end{tikzpicture}
      \caption{
          Our results and hierarchy of the related graph parameters
          (see in \cref{sec:preliminaries} for their definitions).
          For any graph, if the parameter at the tail of an arrow is a constant,
          that is also the case for the one at its head.
          Green indicates that the STC problem is FPT (\cref{thm:distance_to_clique,thm:vi,thm:fes}),
          orange W[1]-hardness (\cref{thm:disjoint-union}),
          and red para-NP-hardness (\cref{thm:hardness:modularwidth}).
      }
      \label{fig:parameters}
  \end{figure}

\paragraph*{Our Contribution.} Our aim in this paper is to present a
clarified and much more detailed picture of how the complexity of {\STC} depends
on treewidth and other notions of graph structure (see \cref{fig:parameters}
for a synopsis of our results).

We begin our work by considering the natural open problem we mentioned above,
namely whether {\STC} is FPT parameterized by treewidth. We answer this
question in the negative and indeed prove something much stronger: Let
$\mathcal{C}$ be any class of graphs that satisfies the (very mild) requirement
that for each integer $i$ there exists a connected graph in $\mathcal{C}$ that
has $i$ vertices. Then, for any such class $\mathcal{C}$, {\STC} is W[1]-hard
parameterized by the vertex deletion distance to a disjoint union of graphs
belonging to $\mathcal{C}$. As a corollary, if we set $\mathcal{C}$ to be the
class of all stars, {\STC} is shown to be W[1]-hard for parameter
vertex-deletion distance to star-forest, hence also for parameter tree-depth
plus feedback vertex number (and consequently also for treewidth). Alternatively,
by setting $\mathcal{C}$ to be the class of all cliques, our proof establishes
W[1]-hardness parameterized by the cluster vertex deletion number, and more
strongly by the twin-cover of the input graph~\citep{dmtcs/Ganian15}. With a
couple of modifications, we then show in \cref{thm:hardness:modularwidth} that
{\STC} remains NP-complete even on graphs of modular-width at most $4$,
linear clique-width $3$, and shrub-depth $2$, improving over the previously
mentioned hardness result by~\cite{jgaa/OkamotoOUU11} for clique-width $3$.
As a matter of fact, the graphs resulting from our
reduction are interval graphs, thus settling the open question posed by~\cite{birthday/Otachi20} on the polynomial-time solvability of {\STC} on
interval graphs in the negative.%
\footnote{This NP-hardness result is in conflict with the polynomial-time
algorithm for interval graphs claimed by~\cite{dam/LinL25a};
a subsequent private communication with the authors indicated that,
in light of \cref{thm:hardness:modularwidth}, they think that their algorithm is incorrect.}

Moving on, we consider the tractability of the problem in graphs of constant
degree.  All previous NP-hardness
results~\citep{algorithmica/BodlaenderFGOL12,algorithmica/LuuC24} require at least one
vertex of unbounded degree. However, assuming that the graph has bounded degree
seems potentially algorithmically useful, as recent work by~\cite{iwoca/Kolman24,stacs/Kolman25}
shows that instances of maximum degree $\Delta$
are amenable to a polynomial-time approximation algorithm of ratio $\bO(\Delta \cdot \log^{3/2} n)$
(this is non-trivial, as the best known ratio on general graphs is $n/2$,
trivially achieved by any spanning tree~\citep{birthday/Otachi20}).
Our next result is to answer an open question posed by~\cite{iwoca/Kolman24} and
show that the problem in fact remains NP-hard even on graphs of degree at most
$8$ (\cref{thm:maxdeg}).  To this end, we make use of a novel gadget based on
grids, simulating the \emph{double-weighted edges} introduced by~\cite{algorithmica/LuuC24}.

Coming back to treewidth, we recall that~\cite{algorithmica/BodlaenderFGOL12}
showed that, when $k$ is
part of the parameter, {\STC} is expressible in MSO$_2$ logic, thus due to
Courcelle's theorem~\citep{iandc/Courcelle90} fixed-parameter tractable by
$\tw+k$. We improve upon this by providing an explicit FPT algorithm of running
time ${(\tw+k)}^{\bO(\tw)} n^{\bO(1)}$. In addition to providing a concrete
reasonable upper-bound on the running time (which cannot be done with
Courcelle's theorem), this explicit algorithm allows us to obtain two further
interesting extensions. First, using a technique introduced by~\cite{icalp/Lampis14},
we develop an efficient FPT approximation scheme
(FPT-AS) when parameterized solely by $\tw$, that is, a $(1 +
\varepsilon)$-approximate algorithm running in time
${(\tw/\varepsilon)}^{\bO(\tw)} n^{\bO(1)}$;
notice that an efficient FPT-AS is the best we can
hope for in this setting, given the W[1]-hardness following
from \cref{thm:disjoint-union}. Second, using a Win/Win argument based on a
result of~\cite{wg/GurskiW00}, we lift our algorithm to also
show an explicit FPT algorithm for the more general parameter $\cw + k$, where
$\cw$ denotes the clique-width of the input graph.

Finally, given all the previously mentioned hardness results, we aim next to
determine which structural parameters \emph{do} render the problem
fixed-parameter tractable.  As a consequence of \cref{thm:disjoint-union}, the
problem remains intractable even on very restricted (dense \emph{and} sparse)
graph classes; we must therefore focus on parameters that evade this hardness
result. We consider three cases: First, the parameter ``distance to clique'' is
not covered by \cref{thm:disjoint-union} because the graph obtained after
removing the deletion set has one component; we show in
\cref{thm:distance_to_clique} that {\STC} is FPT in this case. Second, the
parameter vertex integrity is not covered by \cref{thm:disjoint-union}, as all
components of the graph obtained after removing the deletion set have bounded
size; we show in \cref{thm:vi} that {\STC} is FPT in this case as well, via a
reduction to an ILP with an FPT number of variables. Third, we consider the
parameter feedback edge number, which also falls outside the scope of
\cref{thm:disjoint-union}, and obtain a linear kernel, which leads to
an FPT algorithm with single-exponential parameter dependence for this case.

\paragraph*{Related Work.}
{\STC} was formally introduced by~\cite{dm/Ostrovskii04},
though it had also been previously studied under a different name~\citep{mst/Simonson87}.
There is a plethora of graph-theoretical results in the literature~\citep{dmgt/CastejonO09,dm/Hruska08,dmgt/KozawaO11,dm/KozawaOY09,dm/Law09,dm/LowensteinRR09,dm/Ostrovskii10},
as well as some algorithmic ones~\citep{algorithmica/BodlaenderFGOL12,dm/BodlaenderKMO11,icalp/ChandranCI18,jgaa/OkamotoOUU11}.
See also the survey by~\cite{birthday/Otachi20}.
{\STC} is known to be polynomial-time solvable if $k \le 3$~\citep{algorithmica/BodlaenderFGOL12},
and NP-hard for all fixed $k \ge 5$~\citep{algorithmica/LuuC24}; the case $k=4$ remains open.
\cite{jgaa/OkamotoOUU11} have presented an algorithm running in time $2^n n^{\bO(1)}$,
improving over the brute-force one.
Regarding specific graph classes,
it is known to be polynomial-time solvable for outerplanar graphs~\citep{dm/BodlaenderKMO11},
two-dimensional Hamming graphs~\citep{dmgt/KozawaO11},
complete $k$-partite graphs, and two-dimensional tori~\citep{dm/KozawaOY09}.
On the other hand, it is NP-hard for planar,
split, and chain graphs~\citep{algorithmica/BodlaenderFGOL12,jgaa/OkamotoOUU11}, with the latter result implying NP-hardness for graphs of
clique-width at most $3$.
For fixed $k$, the problem is expressible in MSO$_2$ logic~\citep{algorithmica/BodlaenderFGOL12},
thus due to standard metatheorems~\citep{iandc/Courcelle90} it is FPT by
$\tw+k$. \cite{dm/KozawaOY09} showed a
combinatorial bound (then improved in~\citep{algorithmica/BodlaenderFGOL12})
which proves that for all graphs $G$, $\tw(G)=\bO(\stc(G)\Delta(G))$, where
$\Delta$ denotes the maximum degree of the input graph.  Combining these,
\cite{algorithmica/BodlaenderFGOL12} show that {\STC} is FPT
by $\Delta + k$, and that it is solvable in polynomial time for fixed
$k$ on apex-minor-free graphs.  There are also some results regarding the
problem's
approximability~\citep{algorithmica/BodlaenderFGOL12,iwoca/Kolman24,stacs/Kolman25,algorithmica/LuuC24}.
After the publication of the conference version of this paper~\citep{mfcs/LampisMNOV025},
subsequent works further strengthened some of its hardness results:
\cite{arxiv/Otachi26} showed NP-hardness on proper interval graphs, while
\cite{arxiv/AtaligCDKLSZ26} showed NP-hardness on subcubic graphs.

Finding a spanning tree $T$ of a connected graph
such that $T$ adheres to some constraint,
i.e., $T \in \mathcal{T}$ for some family of trees $\mathcal{T}$,
is an interesting combinatorial question in its own right,
that oftentimes finds applications to other algorithmic problems.
Examples of studied properties include trees of maximum number of
branch or leaf vertices~\citep{siamdm/BonsmaZ11,siamdm/DeBiasioL19,mst/FellowsLMMRS09,icalp/GarganoHSV02,mfcs/GarganoR23,siamdm/KleitmanW91,algorithmica/Lampis12},
of minimum maximum degree~\citep{arxiv/BojikianFGHS25,jacm/SinghL15},
and others~\citep{siamcomp/Agarwal92,dam/BercziKKYY25,jcss/BergougnouxBFGRS25,algorithmica/HalldorssonKMT21,NageleZ19}.
One such important variant of {\STC} is the \textsc{Tree Spanner} problem~\citep{siamdm/CaiC95},
where one asks for a spanning tree of minimum stretch.
The latter has been extensively studied~\citep{siamcomp/AbrahamN19,swat/BorradaileCEMN20,jcss/DraganFG11,siamcomp/ElkinEST08,siamcomp/EmekP08,dam/FeketeK01,ipl/FominGL11},
and the two problems are known to be tightly connected, especially on planar graphs~\citep{algorithmica/LuuC24,birthday/Otachi20}.

Lastly, a closely-related structural graph parameter is the so-called
\emph{edge-cut width}~\citep{wg/BrandCGHK22} or \emph{local feedback edge number}~\citep{nips/GanianK21}.
This is, roughly speaking, the vertex variant of spanning tree congestion,
where one asks to minimize the maximum congestion over the \emph{vertices} of the spanning tree,
where the congestion of a vertex is defined as the number of detours containing it.%
\footnote{Analogously, this has been referred to as the
\emph{vertex remember number} of $G$ relative to $T$ in the literature~\citep[Section~11]{tcs/Bodlaender98}.}
Those parameters have been recently used in the setting of parameterized complexity to show various tractability and incompressibility results~\citep{wg/BrandCGHK22,nips/GanianK21},
and we believe that our work might provide insights into the parameterized (in)tractability of their computation.

\paragraph*{Organization.}
In \cref{sec:preliminaries} we discuss the general preliminaries.
Subsequently, in \cref{sec:hardness} we present the various hardness results,
followed by \cref{sec:treewidth} where we present the explicit FPT algorithm when parameterized by $\tw + k$,
as well as the two results that make use of this.
Moving on, in \cref{sec:fpt_algorithms} we present various fixed-parameter tractability results.
Lastly, in \cref{sec:conclusion} we present the conclusion as well as some directions for future research.


\section{Preliminaries}\label{sec:preliminaries}
Throughout the paper we use standard graph notation~\citep{Diestel17},
and we assume familiarity with the basic notions of parameterized complexity~\citep{books/CyganFKLMPPS15}.
All graphs considered are undirected without loops.
For a graph $G = (V,E)$ and $S \subseteq V$, we denote the \emph{open neighborhood of $S$} by
$N_G(S) = (\bigcup_{s \in S} N_{G}(s)) \setminus S$, while for $s\in V$, we denote the \emph{open neighborhood of $s$} by $N_G(s)=N_G(\{s\})$.
Two vertices $u, v \in V(G)$ are \emph{twins} if $N_G(u) \setminus \{v\} = N_G(v) \setminus \{u\}$;
if additionally $u$ and $v$ are adjacent, we call them \emph{true twins}, otherwise \emph{false twins}.
For $x, y \in \Z$, let $[x, y] = \setdef{z \in \Z}{x \leq z \leq y}$, while $[x] = [1,x]$.

Let $G = (V,E)$ be a connected graph and $T$ a spanning tree of $G$.
The \emph{detour} for $\{u,v\} \in E$ in $T$ is the unique $u$--$v$ path in $T$. Note that the detour of $e \in E(T)$ is $e$ itself.
The \emph{congestion} of $e \in E(T)$, denoted $\cng_{G,T}(e)$, is the number of edges in $G$ whose detours contain~$e$.
In other words, $\cng_{G,T}(e)$ is the size of the fundamental cutset of $T$ defined by~$e$,
that is, $\cng_{G,T}(e) = |E(V(T_{e,1}), V(T_{e,2}))|$,
where $E(X,Y) = \setdef{\{x,y\} \in E}{x \in X, \, y \in Y}$ for disjoint vertex sets $X,Y \subseteq V$
and $T_{e,1}$ and $T_{e,2}$ are the two subtrees of $T$ obtained by cutting~$e$.
The \emph{congestion} of $T$, denoted $\cng_{G}(T)$,
is defined as the maximum over the congestion of all edges in $T$,
i.e., $\cng_{G}(T) = \max_{e \in E(T)} \cng_{G,T}(e)$.
The \emph{spanning tree congestion} of $G$, denoted $\stc(G)$,
is the minimum congestion over all spanning trees of $G$.
Given a connected graph $G$ and an integer $k \in \mathbb{Z}^+$,
{\STC} asks to determine whether $\stc(G) \le k$.

\paragraph*{Graph Parameters.}
We use several standard graph parameters, so we recall here their definitions
and known relations between them.
The \emph{vertex cover} of a graph $G$, denoted by $\vc(G)$, is the size of the smallest
vertex set whose removal destroys all edges.
The \emph{vertex integrity} of a graph $G$, denoted by $\vi(G)$, is the minimum integer $k$
such that there is a vertex set $S \subseteq V(G)$ with $|S| + \max_{C \in \cc(G-S)} |V(C)| \le k$,
where $\cc(G-S)$ denotes the set of connected components in $G-S$.
The \emph{tree-depth} of a graph $G$ can be defined recursively as follows: $\td(K_1)=1$;
if $G$ is disconnected $\td(G)$ is equal to the maximum of the tree-depth of its
connected components; otherwise $\td(G)=\min_{v\in V(G)} \td(G-v)+1$.
A graph $G$ has \emph{feedback vertex} (respectively
\emph{edge}) \emph{number} $k$ if $k$ is the smallest integer so that there exists a set of $k$ vertices (respectively edges) such
that removing them from $G$ destroys all cycles;
we use $\fvs(G)$ and $\fes(G)$ to denote those integers.

Let $G = (V,E)$ be a graph.
The \emph{vertex integrity} of $G$, denoted by $\vi(G)$, is the minimum integer $k$
such that there is a vertex set $S \subseteq V$ with $|S| + \max_{C \in \cc(G-S)} |V(C)| \le k$,
where $\cc(G-S)$ denotes the set of connected components in $G-S$.
The \emph{twin-cover number} of $G$, denoted by $\tc(G)$,
is the size of the smallest vertex set (called \emph{twin-cover}) whose removal results in a
cluster graph (a disjoint union of cliques), with the constraint that each clique is composed of true twins in $G$~\citep{dmtcs/Ganian15}.
If we drop the constraint, this is the \emph{cluster vertex deletion number}~\citep{mfcs/DouchaK12},
denoted by $\cvd(G)$.
The \emph{modular-width} of $G$~\citep{GajarskyLMMO15,GajarskyLO13}
is the smallest integer $k$ such that,
either $|V| \le k$, or $V$ can be partitioned into at most $k' \leq k$ sets $V_1,\ldots,
V_{k'}$, with the following two properties: (i) for all
$i \in [k']$, $V_i$ is a module of $G$,
(ii) for all $i \in [k']$, $G[V_i]$ has modular
width at most $k$. The \emph{distance to clique} of a graph $G$,
denoted by $\dtc(G)$,
is the size of the smallest vertex set whose removal results in a
clique.

The \emph{neighborhood diversity} of a graph $G = (V,E)$,
denoted by $\nd(G)$, is the smallest integer $k$ such that there exists a partition of $V$
into $k$ sets, with each set composed of either only true or only false twins~\citep{algorithmica/Lampis12}.
A module of a graph $G=(V,E)$ is a set of vertices $M \subseteq V$ such that for
all $x \in V \setminus M$ we have that $x$ is either adjacent to all vertices of
$M$ or to none. The \emph{modular-width} of a graph $G=(V,E)$~\citep{GajarskyLMMO15,GajarskyLO13}
is the smallest integer $k$ such that,
either $|V| \le k$, or $V$ can be partitioned into at most $k' \leq k$ sets $V_1,\ldots,
V_{k'}$, with the following two properties: (i) for all
$i \in [k']$, $V_i$ is a module of $G$,
(ii) for all $i \in [k']$, $G[V_i]$ has modular
width at most $k$.

For positive integers $m$ and $d$, we say that a graph $G$ admits a tree-model of $m$ colors and height $d$
if there exists a triple $(T,S,c)$, called \emph{tree-model}, such that $T$ is a rooted tree of height $d$ with
its leaves being the vertices $v$ of $G$ having color $c(v)$ from $[m]$ and each being at distance exactly $d$ from the root,
and $S \subseteq [m] \times [m] \times [d]$, called the \emph{signature} of the tree-model,
such that (i) $(i,j,\ell) \in S$ if and only if $(j,i,\ell) \in S$, and (ii) for any two vertices $u, v \in V(G)$
at distance exactly $2 \ell$ in $T$,
$(c(u),c(v),\ell) \in S$ if and only if $\{ u,v \} \in E(G)$.
A family of graphs $\mathcal{G}$ has \emph{shrub-depth}~\citep{lmcs/GanianHNOM19,mfcs/GanianHNOMR12}
at most $d$ if there exists a positive integer $m$ such that every graph of $\mathcal{G}$ admits a tree-model
with $m$ colors and height $d$.


Let us briefly explain the relations depicted in \cref{fig:parameters}, which
will clarify the results for each parameter.  In most cases, an arrow from
parameter $A$ to parameter $B$ means that for all graphs $G$ we have $B(G) \leq
A(G) +O(1)$. For instance, $\tw(G)\le \fvs(G)+1\le \fes(G)+1$ and $\tw(G)\le
\td(G)\le \vi(G)\le \vc(G)+1$. However, the relation is exponential in a few
cases, namely:  $\nd(G) \le 2^{\vc(G)} + \vc(G)$, $\nd(G) \le 2^{\dtc(G)} +
\dtc(G)$, $\mw(G) \le 2^{\tc(G)} + \tc(G)$, $\cw(G) \le 2^{\tw(G)+1} + 1$.  As
for shrub-depth, its relation with clique-width and tree-depth is
known~\citep[Proposition~3.4]{lmcs/GanianHNOM19}, while it is easy to see that
any class of neighborhood diversity $k$ admits a tree-model of $k$ colors and
height $1$; and that any class of graphs of cluster vertex deletion number $k$
admits a tree-model of $2^k + k$ colors and height $2$.

\section{Hardness Results}\label{sec:hardness}

In this section we present various hardness results for \STC.
We start with showing in \cref{ssec:disjoint-union} that the problem is W[1]-hard
parameterized by the distance to the disjoint union of graphs in $\mathcal{C}$,
for any family of graphs $\mathcal{C}$ that contains connected graphs of any order.
Moving on, in \cref{ssec:modular_width} we adapt our proof and prove NP-hardness for graphs
of modular-width at most $4$.
Finally, in \cref{ssec:maxdeg} we introduce a novel gadget simulating the double-weighted edges
previously used by~\cite{algorithmica/LuuC24}, of which we make use of in order to show
NP-hardness for graphs of constant maximum degree.


\subsection{Distance to Disjoint Union}\label{ssec:disjoint-union}

We start by stating the main theorem of this subsection.

\begin{theorem}\label{thm:disjoint-union}
{\STC} is W[1]-hard parameterized by vertex-deletion distance to disjoint union of graphs in $\mathcal{C}$,
where $\mathcal{C}$ is any graph class such that, for all $p \in \mathbb{Z}^{+}$,
$\mathcal{C}$ contains a connected graph with $p$ vertices which can be generated in time $p^{\bO(1)}$.
\end{theorem}

By taking the set of stars as $\mathcal{C}$,
\cref{thm:disjoint-union} implies the W[1]-hardness parameterized by distance to star forest.

\begin{corollary}\label{cor:star-forest}
{\STC} is W[1]-hard parameterized by distance to star forest (and thus by tree-depth $+$ feedback vertex number).
\end{corollary}

If $\mathcal{C}$ is the set of complete graphs, \cref{thm:disjoint-union} implies W[1]-hardness
parameterized by cluster vertex deletion number~\citep{mfcs/DouchaK12}.
In fact, as we will later see, the proof of \cref{thm:disjoint-union} more strongly implies W[1]-hardness parameterized
by twin-cover number~\citep{dmtcs/Ganian15}.

\begin{corollary}\label{cor:twin-cover}
{\STC} is W[1]-hard parameterized by twin-cover number.
\end{corollary}

For an edge-weighted graph $G = (V,E;\wt)$ with $\wt \colon E \to \Z^{+}$, we define its spanning tree congestion
by setting the congestion of an edge $e$ in a spanning tree $T$ of $G$ as
$\cng_{G,T}(e) = \wt(E(V(T_{e,1}), V(T_{e,2})))$, where $\wt(F) = \sum_{e \in F} \wt(e)$ for $F \subseteq E$.
The following proposition provides a connection between the weighted and the unweighted case.

\begin{proposition}[\cite{algorithmica/BodlaenderFGOL12}]\label{prop:weighted-edge}
Let $G = (V,E;\wt)$ be an edge-weighted graph and $G' = (V',E')$
be an unweighted graph obtained from $G$ by replacing each weighted edge $\{u,v\}$
with $\wt(\{u,v\})$ internally vertex-disjoint $u$--$v$ paths of any lengths,
where one of them may be $\{u,v\}$ itself.
Then, $\stc(G) = \stc(G')$.
\end{proposition}

The rest of the section is dedicated to proving \Cref{thm:disjoint-union}.
The proof follows by a reduction from {\UBP}.
Given unary encoded integers $t, a_{1}, \dots, a_{n} \in \mathbb{Z}^{+}$ with $\sum_{i \in [n]} a_{i} = t B$,
{\UBP} asks whether there is a partition $(A_{1},\dots,A_{t})$ of $[n]$
such that $\sum_{i \in A_{j}} a_{i} = B$ for each $j \in [t]$.
It is known that {\UBP} is W[1]-hard parameterized by $t$~\citep{jcss/JansenKMS13}.
We assume that $t \ge 3$ since otherwise the problem can be solved in polynomial time.
We now proceed to presenting the reduction.

\paragraph*{Construction.}
Let $\mathcal{I} = \langle t; a_{1}, \dots, a_{n} \rangle$ be an instance of {\UBP} with $\sum_{i \in [n]} a_{i} = t B$.
For each $i \in [n]$, let $G_{i} = (V_{i}, E_{i})$ be a connected $a_{i}$-vertex graph belonging to $\mathcal{C}$.
We set $k = 5(t-1)B$ and construct an edge-weighted graph $G = (V,E; \wt)$ as follows.
\begin{enumerate}
  \item Take the disjoint union of all $G_{i}$ for $i \in [n]$.
  \item Add a set of vertices $Q = \setdef{v_{j}}{j \in [t]}$ and all possible edges between $Q$ and $\bigcup_{i \in [n]} V_{i}$.
  \item Add a vertex $r$ and all possible edges between $r$ and $Q$.
  \item Set $\wt(\{r, v_{j}\}) = 3(t-1)B$ for $j \in [t]$, and $\wt(e) = 1$ for all other edges.
\end{enumerate}

\cref{prop:weighted-edge} implies that we can construct, in time polynomial in $n$ and $B$,
an unweighted graph $G'$ from $G$ such that $\stc(G) = \stc(G')$, where
each weighted edge $e$ of weight $\wt(e) \ge 2$ in $G$ is replaced by $\wt(e)$ internally vertex-disjoint paths of length $2$ between the endpoints of~$e$.
Observe that $G' - (\{r\} \cup Q)$ is the disjoint union of all $G_{i}$ and the singleton components
that correspond to the middle vertices of the paths replacing weighted edges.
Since the single-vertex graph belongs to $\mathcal{C}$, $G'$ has distance $|\{r\} \cup Q| = t+1$ to disjoint union of graphs in $\mathcal{C}$.
We can also see that if $\mathcal{C}$ is the set of complete graphs, then $\{r\} \cup Q$ is a twin-cover.
Thus, to prove \cref{thm:disjoint-union} (together with \cref{cor:twin-cover}), it suffices to show that
$\stc(G) \le k$ if and only if $\mathcal{I}$ is a yes-instance of {\UBP}.

\begin{lemma}
If $\mathcal{I}$ is a yes-instance of {\UBP}, then $\stc(G) \le k$.
\end{lemma}

\begin{proof}
Let $(A_{1},\dots,A_{t})$ be a partition of $[n]$
such that $\sum_{i \in A_{j}} a_{i} = B$ for each $j \in [t]$.
We construct a spanning tree $T$ of $G$ by setting
\[
  E(T) = \setdef{\{r, v_{j}\}}{j \in [t]} \cup \bigcup_{j \in [t]} \setdef{\{u, v_{j}\}}{u \in V_{i}, \, i \in A_{j}}.
\]

For each edge $\{u, v_{j}\} \in E(T)$ with $u \in V_{i}$,
since $u$ is a leaf of $T$,
its congestion is $\deg_{G}(u) = t + \deg_{G_{i}}(u) \le t + a_{i} - 1 < k$.

For each $v_{j}$, let $S_{j}$ be the set of vertices of the component of $T - \{r, v_{j}\}$ containing $v_{j}$.
By the construction, we can see that $S_{j} = \{v_{j}\} \cup \bigcup_{i \in A_{j}} V_{i}$.
Thus,
\begin{align*}
  \cng_{G,T}(\{r, v_{j}\})
  &=
  \wt(E(S_{j}, V \setminus S_{j}))
  \\
  &=
  \wt(\{r, v_{j}\})
  + |E(\{v_{j}\}, \; \textstyle\bigcup_{i \in [n] \setminus A_{j}} V_{i})|
  + |E(\textstyle\bigcup_{i \in A_{j}} V_{i}, \; Q \setminus \{v_{j}\})|
  \\
  &=
  3(t-1)B
  + (t-1)B
  + (t-1)B = k.
\end{align*}
This completes the proof.
\end{proof}

\begin{lemma}
If $\stc(G) \le k$,
then $\mathcal{I}$ is a yes-instance of {\UBP}.
\end{lemma}
\begin{proof}
Let $T$ be a spanning tree of $G$ with congestion at most $k$.

We first show that $\{r, v_{j}\} \in E(T)$ for every $j \in [t]$.
Suppose to the contrary that $\{r, v_{j}\} \notin E(T)$ for some $j \in [t]$.
In this case, the $r$--$v_{j}$ path $P$ in $T$ first visits some $v_{j'}$ with $j' \ne j$; i.e., $P = (r, v_{j'}, \dots, v_{j})$.
This implies that the congestion of the edge $\{r, v_{j'}\} \in E(T)$ is at least
$\wt(\{r,v_{j}\}) + \wt(\{r,v_{j'}\}) = 6(t-1)B > k$, a contradiction.

Next we show that for each $i \in [n]$,
there exists exactly one index $j \in [t]$ such that
at least one vertex in $V_{i}$ is adjacent to $v_{j}$ in $T$.

\begin{claim}\label{clm:disjoint-union_partition}
For all $i \in [n]$, there exists $j \in [t]$
such that $N_{T}(V_{i}) \cap Q = \{v_{j}\}$.
\end{claim}
\begin{claimproof}
There is at least one such $j \in [t]$ since $T$ is a spanning tree, i.e., $N_{T}(V_{i}) \cap Q \neq \varnothing$.
Suppose to the contrary that for some $h \in [n]$,
there are two or more vertices in $Q$ that have neighbors in $V_{h}$.
Since $V_{h}$ induces a connected subgraph of $G$ (i.e., $G_{h} = (V_{h}, E_{h})$)
and each edge $\{r, v_{j}\}$ is included in $T$,
there is at least one edge $e_{h} \in E_{h}$ such that the detour for $e_{h}$ in $T$ contains $r$.
Let $R = \sum_{j \in [t]} \cng_{G,T}(\{r, v_{j}\})$.
The edge $e_{h}$ contributes~$2$ to $R$ as its detour passes through two edges incident to $r$.
Each edge $\{r,v_{j}\}$ contributes $\wt(\{r,v_{j}\})$ to $R$.
Now, for $u \in \bigcup_{i \in [n]} V_{i}$,
let $j_{u} \in [t]$ be the index such that $v_{j_{u}}$ appears in the $u$--$r$ path $P_{u,r}$ in $T$.
Since $\{r,v_{j}\} \in E(T)$ for each $j \in [t]$, such $j_{u}$ is unique and $v_{j_{u}}$ appears right before $r$ in $P_{u,r}$.
Observe that for each $j \in [t] \setminus \{j_{u}\}$,
the detour for $\{u, v_{j}\} \in E$ in $T$ consists of $P_{u,r}$ and $v_{j}$,
where $v_{j}$ appears right after $r$.
This detour contributes~$1$ to the congestion of each of the edges $\{r, v_{j_{u}}\}$ and $\{r, v_{j}\}$.
The discussion so far implies that $R > tk$ as follows:
\begin{align*}
  R
  &\ge 2 + \sum_{j \in [t]} \wt(\{r,v_{j}\}) + \sum_{u \in \bigcup_{i \in [n]} V_{i}} 2(t-1) 
  = 2 + 3(t-1)B t + 2(t-1) t B = 2 + t k.
\end{align*}
This contradicts the assumption that each edge of $T$ has congestion at most~$k$ and thus $R \le kt$.
\end{claimproof}

For $j \in [t]$, let $A_{j} = \setdef{i}{\exists u \in V_{i}, \, \{u, v_{j}\} \in E(T)}$.
\cref{clm:disjoint-union_partition} implies that $(A_{1}, \dots, A_{t})$ is a partition of $[n]$.
In particular, the set of vertices of the component of $T - \{r, v_{j}\}$ containing $v_{j}$
is $\{v_{j}\} \cup \bigcup_{i \in A_{j}} V_{i}$.
This implies that
\begin{align*}
  \cng_{G,T}(\{r, v_{j}\})
  &=
  \wt(\{r, v_{j}\})
  + |E(\{v_{j}\}, \; \textstyle\bigcup_{i \in [n] \setminus A_{j}} V_{i})|
  + |E(\textstyle\bigcup_{i \in A_{j}} V_{i}, \; Q \setminus \{v_{j}\})|
  \\
  &=
  3(t-1)B
  + \sum_{i \in [n] \setminus A_{j}} a_{i}
  + (t-1)\sum_{i \in A_{j}} a_{i}
  \\
  &=
  3(t-1)B + tB + (t-2)\sum_{i \in A_{j}} a_{i}.
\end{align*}
Combining this with the assumption $\cng_{G,T}(\{r, v_{j}\}) \le k = 5(t-1)B$,
we obtain that $\sum_{i \in A_{j}} a_{i} \le B$. (Recall that $t \ge 3$.)
Since $\sum_{i \in [n]} a_{i} = tB$, we have $\sum_{i \in A_{j}} a_{i} = B$ for each $j \in [t]$.
\end{proof}

\subsection{Modular-width}\label{ssec:modular_width}

In this subsection we consider {\STC} parameterized by the modular-width of the input
graph~\citep{GajarskyLO13}.
We prove that the problem remains NP-complete even on graphs of modular-width at most $4$;
there is no graph of modular-width $3$ since there is no prime graph with three vertices,
therefore our result leaves open the case of graphs of modular-width at most $2$.
Our result, along with the fact that graphs of modular-width at most $4$ have clique-width at most $3$ (\cref{thm:mw_and_cw}),
slightly improves upon previous work by~\cite{jgaa/OkamotoOUU11},
which showed that the problem is NP-hard for graphs of clique-width at most $3$.
As a matter of fact, the graph we construct has linear clique-width~\citep{siamdm/FellowsRRS09}
at most $3$, thus directly improving over said result.
Nevertheless, \cref{thm:mw_and_cw} is an interesting result in its own right
which we were not able to find in the literature.
Furthermore, the family of graphs constructed by our reduction has shrub-depth at most $2$~\citep{lmcs/GanianHNOM19,mfcs/GanianHNOMR12},
therefore yielding para-NP-hardness for this parameterization as well.
Importantly, the graphs resulting from our reduction are interval graphs,
thus answering a previously posed open question on the polynomial-time solvability of {\STC}
on interval graphs in the negative~\citep{birthday/Otachi20}.
On that note, we also mention a recent follow-up work by~\cite{arxiv/Otachi26},
which shows that {\STC} remains NP-hard even on \emph{proper} interval graphs of clique-width at most $4$.

\begin{theorem}\label{thm:mw_and_cw}
    Any graph of modular-width at most $4$ has clique-width at most $3$.
\end{theorem}

\begin{proof}
    Observe that a graph $G$ with at most four vertices has clique-width at most $3$:
    if $G \ne P_{4}$, then it is a cograph and thus $\cw(G) \le 2$;
    if $G = P_{4}$, then $\cw(G) = 3$ as it is a tree but not a cograph
    (see e.g.,~\citep{dam/CourcelleO00}).

    Let $G$ be a graph of modular-width at most $4$,
    consisting of modules $M_{1}, \ldots, M_{p}$ with $p \le 4$,
    where each module $M_{i}$ induces a graph of clique-width at most $3$.
    Let $H$ be the quotient graph of $G$; that is, $H$ is the graph obtained from $G$ by replacing each module with a single vertex.
    For $i \in [p]$, let $v_{i} \in V(H)$ be the vertex corresponding to $M_{i}$.
    Since $H$ has at most four vertices, it admits a clique-width expression with at most three labels.
    From the expression of $H$, we construct $G$ by replacing the introduction of $v_{i}$ with label $\ell \in [3]$
    with the introduction of $M_{i}$ with label $\ell$.
    (By ``introducing $M_{i}$ with label $\ell$'', we mean the following steps:
    (1) construct $M_{i}$ using at most three labels;
    (2) relabel all vertices in $M_{i}$ by the label $\ell$.)
\end{proof}

In the following, we mostly follow the proof of \cref{thm:disjoint-union} by setting $\mathcal{C}$ to be the class of all complete graphs,
albeit with a few adaptations.
First, the starting point of our reduction is the strongly NP-complete {\ThreeP} problem.
Second, we notice that even though the edge-weighted graph $G$ produced in the proof of \cref{thm:disjoint-union} has
modular-width $2$ when $\mathcal{C}$ is the class of complete graphs (let one module contain the vertices of $Q$ and the other the rest),
that is not the case for the unweighted graph $G'$ produced by \cref{prop:weighted-edge}.
In order to overcome this bottleneck, we emulate the weights of the edges by introducing a sufficiently large clique in our construction.

Given $3n$ unary encoded (not necessarily distinct) integers $a_i \in \mathbb{Z}^{+}$ for $i \in [3n]$,
where $\sum_{i \in [3n]} a_{i} = n B$ and $B/4 < a_i < B/2$ for all $i \in [3n]$,
{\ThreeP} asks whether there is a partition $(A_{1},\ldots,A_{n})$ of $[3n]$
such that $\sum_{i \in A_{j}} a_{i} = B$ for all $j \in [n]$;
notice that the bounds on the values of $a_i$ imply that if $\sum_{i \in A_{j}} a_{i} = B$,
then $|A_j| = 3$.
It is well-known that {\ThreeP} is strongly NP-complete~\citep[SP15]{books/fm/GareyJ79}.

\begin{theorem}\label{thm:hardness:modularwidth}
    {\STC} is NP-complete on interval graphs of modular-width at most $4$,
    linear clique-width at most $3$,
    and shrub-depth at most $2$.
\end{theorem}

\begin{proof}
Let $\mathcal{I} = \langle a_{1}, \dots, a_{3n} \rangle$ be an instance of {\ThreeP} with $\sum_{i \in [3n]} a_i = nB$ and
$B/4 < a_i < B/2$ for all $i \in [3n]$.
Assume without loss of generality that $n \ge 3$.
For each $i \in [3n]$, let $G_{i} = (V_{i}, E_{i})$ be a clique on $a_{i}$ vertices, i.e., $G_i$ is a $K_{a_i}$.
We construct the graph $G$ as follows:
\begin{enumerate}
  \item Take the disjoint union of all $G_{i}$ for $i \in [3n]$.
  \item Add a set of vertices $Q = \setdef{v_{j}}{j \in [n]}$ that form a clique and all possible edges between $Q$ and $\bigcup_{i \in [3n]} V_{i}$.
  \item Add a set of vertices $C = \setdef{c_{z}}{z \in [M]}$ that form a clique and all possible edges between $C$ and $Q$,
        where $M = 4(n-1)B + 2(n-1)$.
  \item Add a vertex $w$, and all possible edges between $w$ and $C$.
\end{enumerate}
This concludes the construction of $G$; we refer to \cref{fig:interval} for an overview of the construction.
We set $k = M + 2(n-1)B + (n-1) = 3M / 2$ and we claim that $\mathcal{I}$ is a yes-instance of {\ThreeP}
if and only if $\stc(G) \le k$.
Before proving this, we provide upper bounds on $G$'s modular-width and linear clique-width.
Furthermore, we prove that the class of graphs produced by our reduction has shrub-depth at most $2$.

\begin{figure}[tbh]
\centering
\includegraphics[scale=0.8]{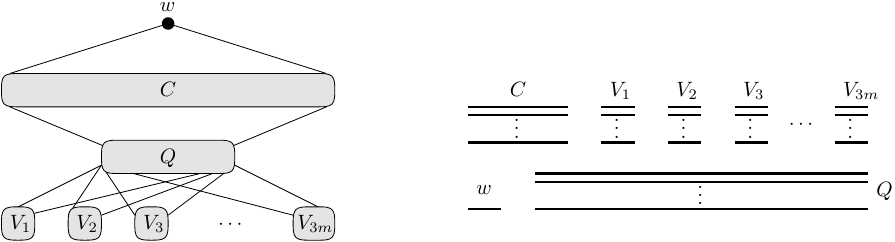}
\caption{The graph $G$ along with an interval representation of it.}
\label{fig:interval}
\end{figure}

\begin{lemma}
    $G$ has modular-width at most $4$ and linear clique-width at most $3$.
\end{lemma}

\begin{proof}
    For the modular-width,
    consider a modular decomposition of $G$ with modules containing vertex sets
    $\{w\}$, $C$, $Q$, and $\bigcup_{i \in [3n]} V_i$ respectively.
    As for the linear clique-width,
    we introduce the vertices of $G$ as follows:
    \begin{itemize}
        \item first introduce%
        \footnote{When we say that we introduce a clique using labels $i$ and $j$,
        we first introduce one of its vertices with label $i$,
        and then as long as there is an unintroduced vertex of the clique,
        we introduce it with label $j$, join $i-j$,
        and relabel $j \to i$.}
        $G_1$ using labels $1$ and $2$,
        \item for $i \in [2,3n]$, (i) introduce $G_i$ using labels $2$ and $3$;
        (ii) relabel $2 \to 1$,
        \item introduce $C$ using labels $2$ and $3$,
        \item introduce vertex $w$ with label $3$,
        and join $2$--$3$,
        \item relabel $2 \to 1$,
        \item for each vertex $v \in Q$,
        (i) introduce $v$ using label $2$;
        (ii) join $1$--$2$; 
        (iii) relabel $2 \to 1$.
    \end{itemize}
    This completes the proof.
\end{proof}

\begin{lemma}
    $G$ admits a tree-model of four colors and depth $2$.
\end{lemma}

\begin{proof}
    We will present a tree-model $(T,S)$ of four colors and depth $2$ for $G$.
    Let $T$ be the rooted tree whose root is connected to the centers of $n+1$ stars, such that
    \begin{itemize}
        \item there is a star per graph $G_i$, whose leaves are $G_i$'s vertices, colored with color $1$,
        \item there is an additional star, whose leaves are the rest of the vertices of the graph,
        namely the vertices belonging to $Q$ and $C$ as well as vertex $w$.
        Color the vertices belonging to $Q$ with color $2$, those belonging to $C$ with color $3$,
        and $w$ with color $4$.
    \end{itemize}
    Finally, let $S = \{ (1,1,1), \, (1,2,2), \, (2,1,2),\, (2,2,1), \, (3,3,1), \, (2,3,1), \, (3,2,1), \, (3,4,1), \, (4,3,1) \}$.
    This completes the proof.
\end{proof}

\begin{lemma}\label{lem:hardness:modular_width:threep->stc}
    If $\mathcal{I}$ is a yes-instance of {\ThreeP}, then $\stc(G) \le k$.
\end{lemma}

\begin{proof}
    Let $(A_{1},\dots,A_{n})$ be a partition of $[3n]$
    such that $|A_j| = 3$ and $\sum_{i \in A_{j}} a_{i} = B$ for all $j \in [n]$.
    We construct a spanning tree $T$ of $G$ by setting
    \[
        E(T) = \setdef{\{c_1, u\}}{u \in (Q \cup C \cup \{w\}) \setminus \{c_1\}}
                \cup \bigcup_{j \in [n]} \setdef{\{u, v_{j}\}}{u \in V_{i}, \, i \in A_{j}}.
    \]

    Each edge $\{u, v_{j}\} \in E(T)$ with $u \in V_{i}$ has congestion
    $\cng_{G,T}(\{u, v_{j}\}) = \deg_{G}(u) = n + a_{i} - 1 < k$ since $u$ is a leaf of $T$.
    Furthermore, all vertices $u \in (C \cup \{w\}) \setminus \{c_1\}$ are leaves in $T$,
    thus $\cng_{G,T}(\{c_1,u\}) = \deg_{G}(u) \le M + n < k$.

    For each $v_{j}$, let $S_{j}$ be the set of vertices of the component of
    $T - \{c_1, v_{j}\}$ containing $v_{j}$.
    By the construction, we can see that $S_{j} = \{v_{j}\} \cup \bigcup_{i \in A_{j}} V_{i}$.
    Thus,
    \begin{align*}
      \cng_{G,T}(\{c_1, v_{j}\})
      &= |E(S_{j}, V(G) \setminus S_{j})|\\
      &= |E(\{v_{j}\}, C)| + |E(\{v_{j}\}, \; \textstyle\bigcup_{i \in [3n] \setminus A_{j}} V_{i})|
      + |E(\textstyle\{v_{j}\} \cup \bigcup_{i \in A_{j}} V_{i}, \; Q \setminus \{v_j\} \})|\\
      &= M + (n-1)B + (n-1) + (n-1)B
      = k.
    \end{align*}
    This completes the proof.
\end{proof}

\begin{lemma}\label{lem:hardness:modular_width:stc->threep}
    If $\stc(G) \le k$, then $\mathcal{I}$ is a yes-instance of {\ThreeP}.
\end{lemma}

\begin{proof}
    Let $T$ be a spanning tree of $G$ with $\cng_{G}(T) \le k$.
    We first prove that $T - (\bigcup_{i \in [3n]} V_i)$ is a star centered on a vertex of $C$.

    \begin{claim}\label{claim:hardness:modular_width:star}
        There exists $\ell \in [M]$ such that $\{c_{\ell}, u\} \in E(T)$,
        for all $u \in (Q \cup C \cup \{w\}) \setminus \{c_{\ell}\}$.
    \end{claim}

    \begin{claimproof}
        Let $Q \cup C \cup \{w\}$ denote the set of \emph{important} vertices of $G$.
        Since $T$ is a tree, there exists a vertex $u \in V(G)$ such that each connected component of $T - u$
        has size at most $|V(G)| / 2$.
        We will first prove that every connected component of $T - u$ contains at most one important vertex.

        Assume otherwise, and let $T_1$ denote a component of $T-u$ containing more than one important vertex,
        with $e$ being the single edge connecting $u$ and $T_1$.
        Let $T_2$ denote the component of $T - e$ containing $u$.
        Notice that $T_1$ and $T_2$ define a partition of the vertices of $C$,
        where $C_i = C \cap V(T_i)$ for $i \in [2]$.

        We show that $2 \leq |C_1| < |C| - 2 = M-2$.
        If $C_1 = \varnothing$, then $\cng_{G,T}(e) \ge 2M > k$, due to the edges between the (at least $2$)
        important vertices of $T_1$ and all vertices of $C_2 = C$.
        If $|C_1| = 1$, then $\cng_{G,T}(e) \ge 2(M-1) > k$ in an analogous way.
        Lastly, $|C_1| \le |V(T_1)| \leq |V(G)| / 2 = (nB + n + M + 1)/2 < M - 2$.

        In that case, it holds that
        \begin{align*}
            \cng_{G,T}(e) &\ge |E(C_1, C_2)|
            =
            |C_1| \cdot (M - |C_1|) \\
            &\ge 2 \cdot (M-2) \\
            &> \frac{3M}{2}
            = k,
        \end{align*}
        a contradiction.

        Consequently, every connected component of $T - u$ contains at most one important vertex.
        In that case, it follows that
        \begin{itemize}
            \item either $u = w$, that is, every connected component of $T-w$ contains at most one important vertex,

            \item or $N_T(w) = \{u\}$, as every neighbor of $w$ is an important vertex.
        \end{itemize}
        Assume the first case, and notice that then, every vertex in the connected component of $T-w$ that contains $v_j \in Q$ belongs to $Q \cup \bigcup_{i \in [3n]} V_i$.
        Since none of these vertices has an edge with $w$, $T$ is disconnected, a contradiction.
        The discussion so far implies that there exists $\ell \in [M]$ such that
        $c_\ell \in C \cap N_T(w)$ denotes the single neighbor of $w$ in $T$,
        with every connected component of $T - c_\ell$ having at most one important vertex.

        To finish our proof,
        assume there exists $v \in (Q \cup C) \setminus \{c_{\ell}\}$ such that $\{c_\ell, v \} \notin E(T)$.
        Consider the connected component of $T-c_\ell$ containing $v$, and notice that apart from $v$,
        it contains only vertices belonging to $\bigcup_{i \in [3n]} V_i$, none of which has an edge with $c_\ell$.
        In that case, $T$ is disconnected, which is a contradiction.
    \end{claimproof}

    Next we show that for each $i \in [3n]$, there exists exactly one index $j \in [n]$ such that
    at least one vertex in $V_{i}$ is adjacent to $v_{j}$ in $T$.
    The proof is analogous to that of \cref{clm:disjoint-union_partition}.

    \begin{claim}\label{claim:hardness:modular_width:partition}
        For all $i \in [3n]$, there exists $j \in [n]$
        such that $N_{T}(V_i) \cap Q = \{v_j\}$.
    \end{claim}

    \begin{claimproof}
        Clearly, $N_{T}(V_i) \cap Q \neq \varnothing$ since $T$ is a spanning tree.
        Assume that there exists $h \in [3n]$ such that $|N_{T}(V_h) \cap Q| \ge 2$,
        with $u^h_1, u^h_2 \in V_h$ denoting the vertices of $V_h$ incident to edges with distinct endpoints in $Q$.
        We remark that due to \cref{claim:hardness:modular_width:star}, $u^h_1$ and $u^h_2$ are distinct vertices of $V_h$.
        Since $V_{h}$ induces a clique and $\{c_{\ell}, v_{j}\} \in E(T)$ for all $j \in [n]$,
        it holds that $e_h = \{u^h_1, u^h_2\} \notin E(T)$,
        meaning that the detour for $e_h$ in $T$ contains $c_{\ell}$.
        Let $R = \sum_{j \in [n]} \cng_{G,T}(\{c_{\ell}, v_{j}\})$.
        Notice that $R \le nk$, since $\cng_{G,T}(e) \le k$ for all $e \in E(T)$.
        The edge $e_{h}$ contributes~$2$ to $R$ as its detour passes through two edges incident to $c_{\ell}$.
        Each edge $\{c_{\ell}, v_{j}\}$ contributes at least $M + n-1$ to $R$,
        because $v_j$ is adjacent to the vertices of $C$ and $Q \setminus \{v_{j}\}$.
        Now, for $u \in \bigcup_{i \in [3n]} V_{i}$,
        let $j_{u} \in [n]$ be the index such that $v_{j_{u}}$ appears in the $u$--$c_{\ell}$ path $P_{u,c_{\ell}}$ in $T$.
        Since $\{c_{\ell},v_{j}\} \in E(T)$ for each $j \in [n]$, such $j_{u}$ is unique and
        $v_{j_{u}}$ appears right before $c_{\ell}$ in $P_{u,c_{\ell}}$.
        Observe that for each $j \in [n] \setminus \{j_{u}\}$,
        the detour for $\{u, v_{j}\} \in E(G)$ in $T$ consists of $P_{u,c_{\ell}}$ and $v_{j}$,
        where $v_{j}$ appears right after $c_{\ell}$.
        This detour contributes~$1$ to the congestion of each of the edges $\{c_{\ell}, v_{j_{u}}\}$
        and $\{c_{\ell},v_{j}\}$.
        The discussion so far implies that
        \[
          R
          \ge 2 + \sum_{j \in [n]} (M+n-1) + \sum_{u \in \bigcup_{i \in [3n]} V_{i}} 2(n-1)
	  = 2 + n (M + n-1) + 2(n-1) n B = 2 + nk
          > nk,
        \]
        which is a contradiction.
    \end{claimproof}

    For $j \in [n]$, let $A_{j} = \setdef{i \in [3n]}{\exists u \in V_{i}, \, \{u, v_{j}\} \in E(T)}$.
    Due to \cref{claim:hardness:modular_width:partition} it follows that
    $(A_{1}, \dots, A_{n})$ is a partition of $[3n]$.
    In particular, the set of vertices of the component of $T - \{c_{\ell}, v_{j}\}$ containing $v_{j}$
    is $\{v_{j}\} \cup \bigcup_{i \in A_{j}} V_{i}$.
    This implies that
    \begin{align*}
      \cng_{G,T}(\{c_{\ell}, v_{j}\})
      &=
      |E(\{v_j\}, C)|
      + |E(\{v_{j}\}, \; \textstyle\bigcup_{i \in [3n] \setminus A_{j}} V_{i})|
      + |E(\textstyle\{v_{j}\} \cup \bigcup_{i \in A_{j}} V_{i}, Q \setminus \{v_j\})|
      \\
      &=
      M + \sum_{i \in [3n] \setminus A_{j}} a_{i}
      + (n-1) + (n-1)\sum_{i \in A_{j}} a_{i}
      \\
      &= M + (n-1) + nB + (n-2)\sum_{i \in A_{j}} a_{i}.
    \end{align*}
    Combining this with the assumption $\cng_{G,T}(c_{\ell}, v_{j}) \le k = M + (n-1) + 2(n-1)B$,
    we obtain that $\sum_{i \in A_{j}} a_{i} \le B$. (Recall that $n \ge 3$.)
    Since $\sum_{i \in [3n]} a_{i} = nB$,
    we have $\sum_{i \in A_{j}} a_{i} = B$ for all $j \in [n]$.
\end{proof}
%
By \cref{lem:hardness:modular_width:threep->stc,lem:hardness:modular_width:stc->threep},
\cref{thm:hardness:modularwidth} follows.
\end{proof}

\subsection{Maximum Degree}\label{ssec:maxdeg}

The last result of \cref{sec:hardness} is the NP-hardness on graphs of constant maximum degree (\cref{thm:maxdeg}).
We note that subsequent work by~\cite{arxiv/AtaligCDKLSZ26}
improves over \cref{thm:maxdeg} by showing that {\STC} is NP-hard even on subcubic graphs.
We first present an alternative gadget for the double-weighted edges
which we subsequently use in our proof.

\paragraph*{Double-weighted edges.}
For the sake of simplicity in our reductions, we will use the concept of double-weighted edges introduced by~\cite{algorithmica/LuuC24}.
A \emph{double edge-weighted graph} $G = (V, E; \wt_{1}, \wt_{2})$ is a graph with two edge-weight functions $\wt_{1}, \wt_{2} \colon E \to \Z^{+}$.
For simplicity, let $\wt(e)$ denote the pair $(\wt_{1}(e), \wt_{2}(e))$ for $e \in E$.
Let $T$ be a spanning tree of $G$.
When considering the congestion of $T$,
the double weights of the edges work slightly differently from the ordinary (single) edge weight considered in \cref{ssec:disjoint-union}.
If $e \notin E(T)$, then it contributes $\wt_{1}(e)$ to the congestions of the edges in the detour for $e$ in $T$;
if $e \in E(T)$, it contributes $\wt_{2}(e)$ to the congestion of itself.
That is, for $e \in E(T)$, it holds that
\[
  \cng_{G,T}(e) = \wt_{1}(E(V(T_{e,1}), V(T_{e,2})) \setminus \{e\}) + \wt_{2}(e).
\]

\cite{algorithmica/LuuC24} showed that for every positive integer $k$,
a double-weighted edge $e$ with $\wt_{1}(e) \le \wt_{2}(e) < k$
can be replaced with a gadget consisting of unweighted edges (i.e., edges $e'$ with $\wt_{1}(e') = \wt_{2}(e') = 1$)
without changing the property of having spanning tree congestion at most~$k$.
Their gadget increases the degree of one endpoint of $e$ by $\wt_{1}(e) - 1$
and the other by $\wt_{1}(e) \cdot (\wt_{2}(e) - \wt_{1}(e) + 1) - 1$.
Because of this increase of the degree, and as in the following reduction proving \cref{thm:maxdeg}, $w_2$ is unbounded, we cannot use their gadget in our proof.

In the following, we present an alternative gadget for double-weighted edges that does not increase the degree too much.
For an integer $n \ge 2$, the $n \times n$ grid is the Cartesian product of two $n$-vertex paths.
We call the degree-$2$ vertices in a grid its \emph{corners}.
It is known that the spanning tree congestion of the $n \times n$ grid is $n$~\citep{dmgt/CastejonO09,dm/Hruska08}.

\begin{lemma}\label{lem:double-weighted_edge}
Let $k$ be a positive integer
and $G = (V,E; \wt_{1}, \wt_{2})$ be a double edge-weighted graph
with an edge $e = \{u,v\} \in E$ satisfying $\wt_{1}(e) < \wt_{2}(e) < k$.
Let $G'$ be the graph obtained from $G$ by the following modification (see \cref{fig:double-weight}~(left)):
\begin{enumerate}
  \item remove $e$;
  \item add $\wt_{1}(e)$ copies of the $(\wt_{2}(e) - \wt_{1}(e) + 1) \times (\wt_{2}(e) - \wt_{1}(e) + 1)$ grid;
  \item for each grid added in the previous step, add edges $\{u,c\}$ and $\{v, c'\}$,
  where $c$ is an arbitrary corner of the grid and $c'$ is the opposite corner (i.e., the corner furthest from $c$);
  \item for each new edge $f \in E(G') \setminus E(G)$, set $\wt_{1}(f) = \wt_{2}(f) = 1$.
\end{enumerate}
Then, $\stc(G) \le k$ if and only if $\stc(G') \le k$.
The degrees of $u$ and $v$ increase by $\wt_{1}(e)-1$ and the maximum degree among newly added vertices is at most~$4$.
\end{lemma}

\begin{figure}[tbh]
\centering
\includegraphics[scale=1]{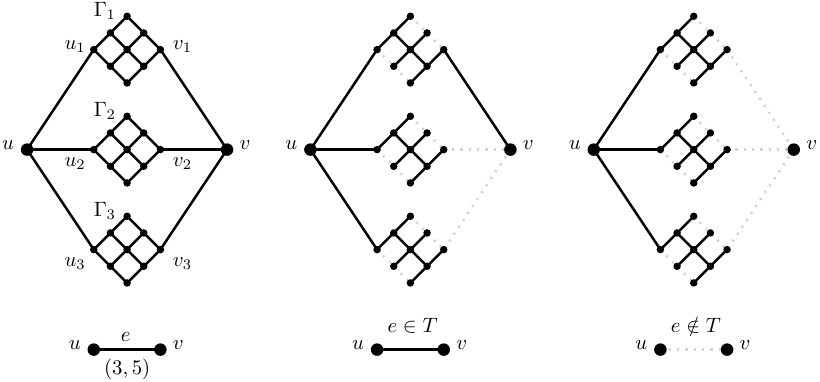}
\caption{(Left) The gadget for a double-weighted edge $e = \{u,v\}$ with $(\wt_{1}(e), \wt_{2}(e)) = (3,5)$.
There are $\wt_{1}(e)$ ($ = 3$) grids connected to $\{u,v\}$ and each grid is of size $(\wt_{2}(e) - \wt_{1}(e) + 1) \times (\wt_{2}(e) - \wt_{1}(e) + 1)$ ($= 3 \times 3$).
(Center \& Right) The intersection of the gadget and the spanning tree $T'$ of $G'$ obtained from a spanning tree $T$ of $G$ for the cases $e \in E(T)$ and $e \notin E(T)$: in the former case, the edges $\{v_1,v\},\{v_2,v\},\{v_3,v\}$ and some edges in $\Gamma_1$ contribute to the congestion of the $u$--$v$ path in $T'$; in the latter, only the edges $\{v_1,v\},\{v_2,v\},\{v_3,v\}$ do.}
\label{fig:double-weight}
\end{figure}

\begin{proof}
The claims about the vertex degrees hold clearly from the construction.
In the following, we show that $\stc(G) \le k$ if and only if $\stc(G') \le k$.

For simplicity, let $\wt(e) = (a,b)$.
We call the $a$ copies of the $(b-a+1) \times (b-a+1)$ grid $\Gamma_{1}, \dots, \Gamma_{a}$.
For $i \in [a]$, let $u_{i}$ and $v_{i}$ denote the neighbors of $u$ and $v$ in $\Gamma_{i}$, respectively.

\proofsubparagraph*{The only-if direction.}
Assume that $\stc(G) \le k$ and $T$ is a spanning tree of $G$ with congestion at most~$k$.
We handle the cases $e \in E(T)$ and $e \notin E(T)$ separately.
In both cases, let $T_{i}$ be a spanning tree of $\Gamma_{i}$ with congestion $b-a+1$.
Such a spanning tree exists as the spanning tree congestion of the $n \times n$ grid is $n$~\citep{dmgt/CastejonO09,dm/Hruska08}.

We first consider the case where $e \in E(T)$.
From $T$, we construct a spanning tree $T'$ of $G'$
by removing $e$,
adding the edges in $T_{i}$ and the edge $\{u, u_{i}\}$ for $i \in [a]$,
and adding the edge $\{v, v_{1}\}$ (see \cref{fig:double-weight}~(center));
that is,
\[
  E(T') = (E(T) \setminus \{e\}) \cup \bigcup_{i \in [a]} (E(T_{i}) \cup \{\{u, u_{i}\}\}) \cup \{\{v, v_{1}\}\}.
\]
We can see that $T'$ is a spanning tree of $G'$.
Observe that the edges in $E(G') \setminus E(G)$ do not use the edges in $E(T) \setminus E(T')$ in their detours in $T'$.
Observe also that for each edge $f \in E(G) \setminus \{e\}$, if the detour $P$ for $f$ in $T$ includes $e$,
then the detour for $f$ in $T'$ is obtained from $P$ by replacing $e$ with the $u$--$v$ path in $T'$.

Let $f \in E(T')$.
If $f \in E(T) \setminus \{e\}$, then $\cng_{G',T'}(f) = \cng_{G,T}(f)$.
In the following, assume that $f \notin E(T)$.
If $f = \{u, u_{i}\}$ for some $i \in [2,a]$, then $\cng_{G',T'}(f) = |\{\{u, u_{i}\}, \{v, v_{i}\}\}| = 2 \le k$.
If $f \in E(\Gamma_{i})$ for some $i \in [a]$ and $f$ does not belong to the $u$--$v$ path in $T'$,
then $\cng_{G',T'}(f) \le \cng_{\Gamma_{i},T_{i}}(f) + 1 \le (b-a+1) + 1 \le k$, where the term $+1$ may appear by the detour for $\{v, v_{i}\}$ when $i \ne 1$.
In the remaining cases, $f$ belongs to the $u$--$v$ path in $T'$.
If $f$ is $\{u, u_{1}\}$ or $\{v, v_{1}\}$, then
\[
  \cng_{G',T'}(f) =
  \wt_{1}(E(V(T_{e,1}), V(T_{e,2})) \setminus \{e\}) + |\{f\} \cup \setdef{\{v, v_{i}\}}{i \in [2,a]}|.
\]
Similarly, if $f \in E(\Gamma_{1})$, then
\[
  \cng_{G',T'}(f) \le
  \wt_{1}(E(V(T_{e,1}), V(T_{e,2})) \setminus \{e\}) + |\setdef{\{v, v_{i}\}}{i \in [2,a]}| + \cng_{\Gamma_{1},T_{1}}(f).
\]
In both cases, $\cng_{G',T'}(f) \le \cng_{G,T}(e) \le k$ holds
since $|\setdef{\{v, v_{i}\}}{i \in [2,a]}| + \cng_{\Gamma_{1},T_{1}}(f) \le (a-1) + (b-a+1) = b = \wt_{2}(e)$.

Next, consider the case where $e \notin E(T)$.
From $T$, we construct a spanning tree $T'$ of $G'$
by adding the edges in $T_{i}$ and the edge $\{u, u_{i}\}$ for $i \in [a]$ (see \cref{fig:double-weight}~(right));
that is,
\[
 E(T') = E(T) \cup \bigcup_{i \in [a]} (E(T_{i}) \cup \{\{u, u_{i}\}\}).
\]
Clearly, $T'$ is a spanning tree of $G'$.
The detours for the $a$ edges $\{v, v_{i}\}$ ($i \in [a]$) contain the $u$--$v$ path in $T$.
No other edge in $E(T') \setminus E(T)$ contributes to the congestions of edges in $E(T)$.
Thus, if $f \in E(T)$ and $f$ does not belong to the $u$--$v$ path in $T$, then
$\cng_{G',T'}(f) = \cng_{G,T}(f) \le k$.
If $f \in E(T)$ and $f$ belongs to the $u$--$v$ path in $T$, then
$\cng_{G',T'}(f) = \cng_{G,T}(f) - \wt_{1}(e) + |\setdef{\{v, v_{i}\}}{i \in [a]}| = \cng_{G,T}(f) \le k$.
Finally, we consider the case where $f \in E(T') \setminus E(T)$.
If $f$ is $\{u,u_{i}\}$ for some $i \in [a]$, then $\cng_{G',T'}(f) = |\{\{u,u_{i}\}, \{v,v_{i}\}\}| = 2 \le k$.
Otherwise, $f \in E(\Gamma_{i})$ for some $i \in [a]$ and $\cng_{G',T'}(f) \le \cng_{\Gamma_{i},T_{i}}(f) + 1 \le (b-a+1) + 1 \le b+1 \le k$,
where the term $+1$ is necessary if $f$ is on the $u_{i}$--$v_{i}$ path.

\proofsubparagraph*{The if direction.}
Assume that $\stc(G') \le k$ and $T'$ is a spanning tree of $G'$ with congestion at most~$k$.
Since the vertex set $\{u,v\}$ is a vertex-cut that separates the gadget replacing $e$ and the rest,
the $u$--$v$ path $P$ in $T'$ is either completely contained in the gadget or not intersecting it at all.

Let us first consider the case where the $u$--$v$ path $P$ is contained in the gadget.
From $T'$, we construct a spanning tree $T$ of $G$ by removing the gadget and then adding $e$;
that is, $E(T) = (E(T') \cap E(G)) \cup \{e\}$.

For each $f \in E(G)$, if the detour for $f$ in $T'$ contains $P$,
then its detour in $T$ is obtained from the one in $T'$ by replacing $P$ with $e$;
otherwise, the detour for $f$ in $T$ is the same as the one in $T'$.
This implies that $\cng_{G,T}(f) = \cng_{G',T'}(f)$ for $f \in E(T) \setminus \{e\}$.

Now we show that $\cng_{G,T}(e) \le k$.
There is exactly one index $i \in [a]$ such that $P$ contains both $\{u, u_{i}\}, \{v, v_{i}\}$ and a $u_{i}$--$v_{i}$ path.
It is known that in every spanning tree of the $n \times n$ grid,
the path from a corner to the opposite corner contains an edge with congestion at least~$n$~\citep{dm/Hruska08}.
This implies that there is an edge $e_{i}$ in $P$
such that at least $b-a+1$ edges in $\Gamma_{i}$ use $e_{i}$ in their detours in $T'$.
For each $j \in [a] \setminus \{i\}$, there is a $u$--$v$ path whose internal vertices belong to $\Gamma_j$,
and at least one edge in this path uses $P$ in its detour.
In total, at least $b$ edges in the gadget use $e_{i}$ in their detours.
Furthermore, if an edge $f \in E(G) \setminus \{e\}$ uses $e$ in its detour in $T$,
then $f$ uses $P$ (and thus $e_{i}$) in its detour in $T'$. Thus we have
\begin{align*}
  \cng_{G,T}(e)
  &=
  \wt_{1}(E(V(T_{e,1}), V(T_{e,2})) \setminus \{e\}) + \wt_{2}(e)
  \\
  &=
  \wt_{1}(E(V(T_{e,1}), V(T_{e,2})) \setminus \{e\}) + b
  \le
  \cng_{G',T'}(e_{i}) \le k.
\end{align*}

Next we consider the case where the $u$--$v$ path $P$ does not intersect the gadget.
From $T'$, we construct a spanning tree $T$ of $G$ by removing the gadget;
that is, $E(T) = E(T') \cap E(G)$.
If $f \in E(T) \setminus E(P)$,
the edges using $f$ for their detours are the same in $T$ and $T'$,
and thus $\cng_{G,T}(f) = \cng_{G',T'}(f) \le k$.

Now assume that $f \in E(P)$.
If an edge in $G$ uses $f$ in its detour in $T$, then it uses $f$ in its detour in $T'$ as well.
For each $i \in [a]$, one of $\{u, u_{j}\}$ and $\{v, v_{j}\}$ uses $P$ (and thus $f$) in its detour in $T'$.
On the other hand, $e$ uses $P$ in its detour in $T$.
Thus we have
\begin{align*}
  \cng_{G,T}(f)
  &=
  \wt_{1}(E(V(T_{f,1}), V(T_{f,2})) \setminus \{e,f\}) + \wt_{1}(e) + \wt_{2}(f)
  \\
  &=
  \wt_{1}(E(V(T_{f,1}), V(T_{f,2})) \setminus \{e,f\}) + a + \wt_{2}(f)
  \le
  \cng_{G',T'}(f) \le k.
\end{align*}
This completes the proof.
\end{proof}

The problem {\BSAT} (also appearing in the literature as \textsc{2P2N-3SAT}) is a restricted version of 3-SAT\@:
an instance of {\BSAT} consists of a set $X$ of $n$ variables and a set $C$ of $m$ clauses such that
each clause has exactly three literals corresponding to three different variables
and each variable appears exactly twice positively and exactly twice negatively.
It is known that {\BSAT} is NP-complete~\citep{eccc/ECCC-TR03-049},
even if the formula contains only monotone clauses~\citep{dam/DarmannD21}.

\begin{theorem}\label{thm:maxdeg}
  {\STC} is NP-complete on graphs of maximum degree at most~$8$.
\end{theorem}

\paragraph*{Construction.}
Let $(X,C)$ be an instance of {\BSAT} with $X = \{x_{1}, \dots, x_{n}\}$ and $C = \{c_{1}, \dots, c_{m}\}$.
We assume that $m \ge 3$ (otherwise the problem becomes trivial).
Set $k = 2m+3$ ($\ge 9$).
From $(X,C)$, we construct a double edge-weighted graph $G = (V,E; \wt_{1}, \wt_{2})$ as follows (see \cref{fig:bsat}).
\begin{itemize}
  \item For $i \in [n]$, take a cycle $(x_{i}, y_{i}, \bar{x}_{i}, z_{i})$ of four new vertices.
  \begin{itemize}
    \item Set $\wt(\{x_{i}, y_{i}\}) = \wt(\{\bar{x}_{i}, y_{i}\}) =  (4, k-3)$ and $\wt(\{x_{i}, z_{i}\}) = \wt(\{\bar{x}_{i}, z_{i}\}) = (1, 1)$.
  \end{itemize}
  \item For $i \in [n-1]$, add the edge $\{z_{i}, z_{i+1}\}$, thus forming the path $(z_{1}, \dots, z_{n})$.
  \begin{itemize}
    \item Set $\wt(\{z_{i}, z_{i+1}\}) = (3,3)$.
  \end{itemize}
  \item For $j \in [m]$, take a new vertex $c_{j}$.
  \item For $i \in [n]$ and $j \in [m]$, add the edge $\{x_{i}, c_{j}\}$ (resp.\ $\{\bar{x}_{i}, c_{j}\}$) if $x_{i} \in c_{j}$ (resp.\ $\bar{x}_{i} \in c_{j}$).
  \begin{itemize}
    \item Set $\wt(\{x_{i}, c_{j}\}) = (1,k-2)$ (resp.\ $\wt(\{\bar{x}_{i}, c_{j}\}) = (1,k-2)$) if the edge exists.
  \end{itemize}
\end{itemize}

\begin{figure}[tbh]
\centering
\includegraphics[scale=1]{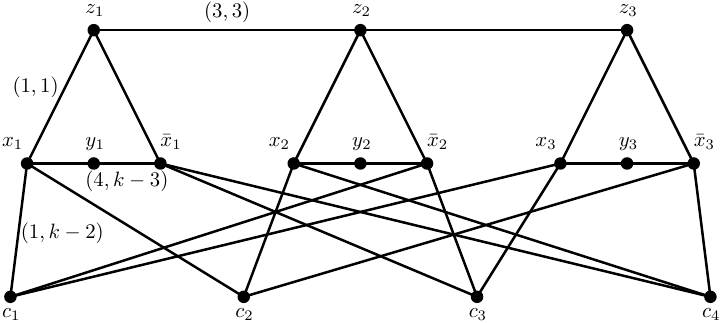}
\caption{The construction of $G$, where
$c_{1} = \{x_{1}, \bar{x}_{2}, x_{3}\}$,
$c_{2} = \{x_{1}, x_{2}, \bar{x}_{3}\}$,
$c_{3} = \{\bar{x}_{1}, \bar{x}_{2}, x_{3}\}$, and
$c_{4} = \{\bar{x}_{1}, x_{2}, \bar{x}_{3}\}$.}
\label{fig:bsat}
\end{figure}

To prove \cref{thm:maxdeg}, it suffices to prove the following two claims.
\begin{enumerate}
  \item  In polynomial time, we can construct an unweighted graph $G'$ such that
  \begin{itemize}
    \item $\stc(G) \le k$ if and only if $\stc(G') \le k$;
    \item the maximum degree of $G'$ is at most~$8$.
  \end{itemize}
  \item $\stc(G) \le k$ if and only if $(X,C)$ is a yes-instance of {\BSAT}.
\end{enumerate}
\cref{lem:maxdeg_unweighting} shows the first claim
and \cref{lem:maxdeg_equivalence} shows the second one.

\begin{lemma}\label{lem:maxdeg_unweighting}
In time polynomial in $m+n$, one can construct an unweighted graph $G' = (V',E')$ from the double edge-weighted graph $G = (V, E; \wt_{1}, \wt_{2})$
with maximum degree at most~$8$
such that $\stc(G) \le k$ if and only if $\stc(G') \le k$.
\end{lemma}
\begin{proof}
We keep the edges with weight $(1,1)$.
For each edge $\{u,v\}$ with weight $(3,3)$,
we use \cref{prop:weighted-edge} and replace the edge $\{u,v\}$ with three internally vertex disjoint $u$--$v$ paths of length~$2$.
For each edge with weight $(4, k-3)$ or $(1, k-2)$, we replace the edge with the gadget in \cref{lem:double-weighted_edge}.
We call the obtained unweighted graph $G' = (V', E')$.
Since $k = 2m + 3$, the construction can be done in time polynomial in $m+n$.
By \cref{prop:weighted-edge,lem:double-weighted_edge}, $\stc(G) \le k$ if and only if $\stc(G') \le k$.

Now we show that $G'$ has maximum degree at most~$8$.
Observe that $V' \supseteq V$ and that the maximum degree of the vertices in $V' \setminus V$ is at most~$4$.
For each $i \in [n]$,
$\deg_{G'}(z_{i}) \le \deg_{G}(z_{i}) + 4 \le 8$,
$\deg_{G'}(y_{i}) = \deg_{G}(y_{i}) + 6 = 8$,
$\deg_{G'}(x_{i}) = \deg_{G}(x_{i}) + 3 = 7$, and
$\deg_{G'}(\bar{x}_{i}) = \deg_{G}(\bar{x}_{i}) + 3 = 7$.
For each $j \in [m]$,
$\deg_{G'}(c_{j}) = \deg_{G}(c_{j}) = 3$.
\end{proof}

\begin{lemma}\label{lem:maxdeg_equivalence}
$\stc(G) \le k$ if and only if $(X,C)$ is a yes-instance of {\BSAT}.
\end{lemma}

\begin{proof}
To show the if direction,
assume that $(X,C)$ is a yes-instance of {\BSAT}
and $\alpha \colon X \to \{0,1\}$ is a satisfying assignment of $(X,C)$.
From $\alpha$, we construct a spanning tree $T$ of $G$ as follows (see \cref{fig:bsat_assignment}).
\begin{itemize}
  \item Take all edges in the path $(z_{1}, \dots, z_{n})$.
  \item For $i \in [n]$, take both edges incident to $y_{i}$.
  \item For $i \in [n]$, if $\alpha(x_{i}) = 1$, then take the edge $\{x_{i}, z_{i}\}$;
  otherwise, take the edge $\{\bar{x}_{i}, z_{i}\}$;

  \item For $j \in [m]$, fix one literal $\ell_{i} \in c_{j}$ evaluated to~$1$ under $\alpha$
  (that is, if $\ell_{i} = x_{i}$, then $\alpha(x_{i}) = 1$; otherwise $\alpha(x_{i}) = 0$)
  and take the edge $\{c_{j}, \ell_{i}\}$.
\end{itemize}

We first show that $\cng_{G,T}(\{z_{i}, z_{i+1}\}) \le k$ for $i \in [n-1]$.
Observe that an edge $e \in E(G) \setminus E(T)$ uses $\{z_{i}, z_{i+1}\}$ in its detour only if $e$ is incident to some $c_{j}$.
In total, there are exactly $3m$ edges incident to some $c_{j}$ and $m$ of them are included in $T$.
Therefore, $\cng_{G,T}(\{z_{i}, z_{i+1}\}) \le 2m + \wt_{2}(\{z_{i}, z_{i+1}\}) = 2m+3 = k$.

We next show that, for $i \in [n]$, each edge incident to $y_{i}$ has congestion~$k$.
By symmetry, assume that $\{x_{i}, z_{i}\} \in E(T)$.
Let $c_{j_{1}}, c_{j_{2}}$ be the clauses that $\bar{x}_{i}$ appears in.
Then we have
\begin{math}
\cng_{G,T}(\{\bar{x}_{i}, y_{i}\})
= \wt_{1}(\{\{\bar{x}_{i}, z_{i}\}, \{\bar{x}_{i}, c_{j_{1}}\}, \{\bar{x}_{i}, c_{j_{2}}\}\}) + \wt_{2}(\{\bar{x}_{i}, y_{i}\})
= 3 + (k-3) = k
\end{math}.
Similarly,
\begin{math}
\cng_{G,T}(\{x_{i}, y_{i}\})
= \wt_{1}(\{\{\bar{x}_{i}, z_{i}\}, \{\bar{x}_{i}, c_{j_{1}}\}, \{\bar{x}_{i}, c_{j_{2}}\}\}) + \wt_{2}(\{x_{i}, y_{i}\})
= 3 + (k-3) = k
\end{math}.

Now we show that, for $j \in [m]$, the edge incident to $c_{j}$ has congestion~$k$.
Let $c_{j} = \{\ell_{i_{1}}, \ell_{i_{2}}, \ell_{i_{3}}\}$, where $\{c_{j}, \ell_{i_{1}}\} \in E(T)$.
Since $c_{j}$ is a leaf of $T$, it holds that
\begin{math}
\cng_{G,T}(\{c_{j}, \ell_{i_{1}}\})
= \wt_{1}(\{c_{j}, \ell_{i_{2}}\}, \{c_{j}, \ell_{i_{3}}\}) + \wt_{2}(\{c_{j}, \ell_{i_{1}}\})
= 2 + (k-2) = k
\end{math}.

Finally, we show that, for $i \in [n]$, the edge in $T$ between $z_{i}$ and $x_{i}$ or $\bar{x}_{i}$ has congestion at most~$8$ ($\le k$).
Assume by symmetry that $\{z_{i}, x_{i}\} \in E(T)$.
Let $T_{x_{i}}$ be the connected component of $T - \{z_{i}, x_{i}\}$ including $x_{i}$.
The vertex set $V(T_{x_{i}})$ includes $x_{i}, y_{i}, \bar{x}_{i}$ and zero, one, or two vertices in $\{c_{1}, \dots, c_{m}\}$ adjacent to $x_{i}$.
Thus, we have $|E(V(T_{x_{i}}), V(G) \setminus V(T_{x_{i}}))| \le 8$.
Since every edge $e$ in $E(V(T_{x_{i}}), V(G) \setminus V(T_{x_{i}}))$ satisfies $\wt_{1}(e) = 1$, it holds that
$\cng_{G,T}(\{z_{i}, x_{i}\}) = \wt_{1}(E(V(T_{x_{i}}), V(G) \setminus V(T_{x_{i}})) \setminus \{\{z_{i}, x_{i}\}\}) + \wt_{2}(\{z_{i}, x_{i}\}) \le 8$.
\begin{figure}[tbh]
\centering
\includegraphics[scale=1]{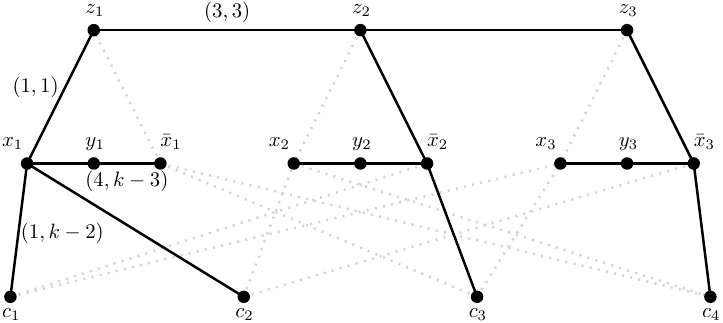}
\caption{The construction of $T$, where $\alpha = \{(x_{1}, 1), (x_{2}, 0), (x_{3}, 0)\}$.}
\label{fig:bsat_assignment}
\end{figure}

To show the only-if direction, assume that $\stc(G) \le k$
and $T$ is a spanning tree of $G$ with congestion at most~$k$.
We first show the following series of claims.

\begin{claim}\label{clm:maxdeg_y-edges}
For $i \in [n]$, $\{x_{i}, y_{i}\}, \{\bar{x}_{i}, y_{i}\} \in E(T)$.
\end{claim}

\begin{claimproof}
Since $\deg_{G}(y_{i}) = 2$, at least one of $\{x_{i}, y_{i}\}$ and $\{\bar{x}_{i}, y_{i}\}$ belongs to $T$.
Suppose to the contrary that one of them does not belong to $T$.
By symmetry, we may assume that $\{x_{i}, y_{i}\} \notin E(T)$, and thus $\{\bar{x}_{i}, y_{i}\} \in E(T)$.
Then, we have $\cng_{G,T}(\{\bar{x}_{i}, y_{i}\}) = \wt_{1}(\{x_{i}, y_{i}\}) + \wt_{2}(\{\bar{x}_{i}, y_{i}\}) = 4 + k - 3 > k$, a contradiction.
\end{claimproof}

\begin{claim}\label{clm:maxdeg_z-edges}
For $i \in [n-1]$, $\{z_{i}, z_{i+1}\} \in E(T)$.
\end{claim}

\begin{claimproof}
Suppose to the contrary that $\{z_{i}, z_{i+1}\} \notin E(T)$ for some $i \in [n-1]$.
Let $P$ be the $z_{i}$--$z_{i+1}$ path in $T$.
Observe that $G - \{c_{1},\dots,c_{m}\}$ is a cactus graph that consists of the path $(z_{1},\dots,z_{n})$ and $4$-cycles attached to the vertices on the path.
Therefore, the edge $\{z_{i}, z_{i+1}\}$ is the unique $z_{i}$--$z_{i+1}$ path in $G - \{c_{1},\dots,c_{m}\}$.
This implies that $P$ contains $c_{j}$ for some $j \in [m]$.
Let $e \in E(P)$ be an edge incident to $c_j$.
Then, we have $\cng_{G,T}(e) \ge \wt_{1}(\{z_{i}, z_{i+1}\}) + \wt_{2}(e) = 3 + k-2 > k$, a contradiction.
\end{claimproof}

\begin{claim}\label{clm:maxdeg_x-edges}
For $i \in [n]$, exactly one of $\{x_{i}, z_{i}\}$ and $\{\bar{x}_{i}, z_{i}\}$ belongs to $T$.
\end{claim}

\begin{claimproof}
By \cref{clm:maxdeg_y-edges}, at most one of $\{x_{i}, z_{i}\}$ and $\{\bar{x}_{i}, z_{i}\}$ may belong to $T$.
Suppose to the contrary that none of $\{x_{i}, z_{i}\}$ and $\{\bar{x}_{i}, z_{i}\}$ belong to $T$.
Let $\{y_{i}, \ell_{i}\}$ and $\{\ell_{i}, c_{j}\}$ be the first and second edges, respectively,
in the $y_{i}$--$z_{i}$ path in $T$, where $\ell_{i} \in \{x_{i}, \bar{x}_{i}\}$.
Observe that the edge set $\{\{x_{i}, z_{i}\}, \{\bar{x}_{i}, z_{i}\}, \{\ell_{i}, c_{j}\}\}$ is not an edge-cut in $G$:
the graph $G - \{x_{i}, y_{i}, \bar{x}_{i}\}$ is connected; and there is an edge not included in $\{\{x_{i}, z_{i}\}, \{\bar{x}_{i}, z_{i}\}, \{\ell_{i}, c_{j}\}\}$ that connects $G - \{x_{i}, y_{i}, \bar{x}_{i}\}$ and $G[\{x_{i}, y_{i}, \bar{x}_{i}\}]$
(e.g., there is an edge between $\ell_{i}$ and $c_{j'}$ for some $j' \ne j$).
Therefore, it holds that the edge set $E(V(T_{\{\ell_{i}, c_{j}\},1}), V(T_{\{\ell_{i}, c_{j}\},2}))$ is
a proper super set of $\{\{x_{i}, z_{i}\}, \{\bar{x}_{i}, z_{i}\}, \{\ell_{i}, c_{j}\}\}$.
In that case,
\begin{align*}
  \cng_{G,T}(\{\ell_{i}, c_{j}\})
  &>
  \wt_{1}(\{\{x_{i}, z_{i}\}, \{\bar{x}_{i}, z_{i}\}\})
  +
  \wt_{2}(\{\ell_{i}, c_{j}\})
  = 2 + (k-2) = k,
\end{align*}
which is a contradiction.
\end{claimproof}

By the claims above, $T - \{c_{1}, \dots, c_{m}\}$ is connected.
Since $\{c_{1}, \dots, c_{m}\}$ is an independent set, we can conclude that each $c_{j}$ is a leaf of $T$.

\begin{claim}\label{clm:maxdeg_c-leaves}
For $j \in [m]$, $c_{j}$ is a leaf of $T$.
\end{claim}

\begin{claim}\label{clm:maxdeg_assignment}
Let $i \in [n]$ and $\ell_{i} \in \{x_{i}, \bar{x}_{i}\}$.
If $\{\ell_{i}, c_{j}\} \in E(T)$ for some $j \in [m]$, then $\{\ell_{i}, z_{i}\} \in E(T)$.
\end{claim}

\begin{claimproof}
By symmetry, it suffices to show that
if $\{x_{i}, c_{j}\} \in E(T)$ for some $j \in [m]$, then $\{x_{i}, z_{i}\} \in E(T)$.
Suppose to the contrary that
$\{x_{i}, c_{j}\} \in E(T)$ for some $j \in [m]$ but $\{x_{i}, z_{i}\} \notin E(T)$.
By \cref{clm:maxdeg_y-edges}, $\{x_{i}, y_{i}\}, \{\bar{x}_{i}, y_{i}\} \in E(T)$.
By \cref{clm:maxdeg_x-edges}, $\{\bar{x}_{i}, z_{i}\} \in E(T)$.
Observe that the edge $\{x_{i}, z_{i}\}$ and the edges incident to $c_{j}$ except for $\{x_{i}, c_{j}\}$ use $\{x_{i}, y_{i}\}$ in their detours in $T$.
Let $j' \in [m]$ be the index such that $j' \ne j$ and $\{x_{i}, c_{j'}\} \in E(G)$.
If $\{x_{i}, c_{j'}\} \notin E(T)$, then $\{x_{i}, c_{j'}\}$ uses $\{x_{i}, y_{i}\}$ in its detour in $T$.
Otherwise, the edges incident to $c_{j'}$ except for $\{x_{i}, c_{j'}\}$ use $\{x_{i}, y_{i}\}$ in their detours in $T$.
In total, at least four edges not including $\{x_{i}, y_{i}\}$ itself use $\{x_{i}, y_{i}\}$ in their detours in $T$.
Therefore, we have
$\cng_{G,T}(\{x_{i}, y_{i}\}) \ge 4 + \wt_{2}(\{x_{i}, y_{i}\}) = 4 + (k-3) > k$, a contradiction.
\end{claimproof}

Now we construct an assignment $\alpha \colon X \to \{0,1\}$ from $T$
by setting $\alpha(x_{i}) = 1$ if $\{x_{i}, z_{i}\} \in E(T)$
and $\alpha(x_{i}) = 0$ if $\{\bar{x}_{i}, z_{i}\} \in E(T)$.
Note that, by \cref{clm:maxdeg_x-edges}, exactly one of these cases happens for each $i \in [n]$.
We show that $\alpha$ is a satisfying assignment of $(X,C)$.
Let $j \in [m]$.
There is an index $i \in [n]$ such that $x_{i}$ or $\bar{x}_{i}$ is adjacent to $c_{j}$ in $T$.
By symmetry, assume that $\{x_{i}, c_{j}\} \in E(T)$.
By \cref{clm:maxdeg_assignment}, $\{x_{i}, z_{i}\} \in E(T)$ holds.
Therefore, we have $\alpha(x_{i}) = 1$, and thus the clause $c_{j}$ is satisfied by $\alpha$.
\end{proof}




\section{Algorithms for Bounded Treewidth}\label{sec:treewidth}

In this section, we take a second look at the complexity of \STC\ parameterized by treewidth, a problem shown to be W[1]-hard by combining \cref{cor:star-forest} with the fact that graphs of bounded tree-depth also have bounded treewidth.
One way to deal with this hardness is to consider additional parameters, so we
begin by presenting in \cref{sec:fpt-tw-k} an FPT algorithm parameterized by
treewidth plus the desired congestion $k$. Our algorithm follows the standard
technique of performing dynamic programming over a tree decomposition, though
with a few necessary tweaks (informally, we have to guess the general structure
of the spanning tree, including parts of the graph that will appear ``in the
future'').

We note here that the fact that \STC\ is FPT for this parameterization was already shown by~\cite{algorithmica/BodlaenderFGOL12}, who proved that if $k$ is a parameter, then the problem is MSO$_2$-expressible, hence solvable via Courcelle's theorem.
Nevertheless, we are still motivated to provide an
explicit algorithm for the parameter ``treewidth plus $k$'' for several
reasons.

First, using Courcelle's theorem does not provide any usable upper
bound on the running time, while we show our algorithm to run in
${(k+w)}^{\bO(w)}n^{\bO(1)}$, where $w$ is the treewidth; this implies, for instance,
that \STC\ is in XP parameterized by treewidth alone (as $\stc(G)\le n^2$ for
all $G$), a fact that cannot be inferred using Courcelle's theorem.

Second, and more importantly, having an explicit algorithm at hand, we are able
to obtain an answer to the following natural question: given that solving \STC\
is hard parameterized by treewidth, is there an FPT algorithm that closely
approximates the optimal congestion? By applying a technique introduced by~\cite{icalp/Lampis14}
which modifies exact DP algorithms to obtain approximate
ones, we get an FPT approximation scheme, which runs in time
${(\frac{w}{\varepsilon})}^{\bO(w)}n^{\bO(1)}$ (that is, FPT in
$w+\frac{1}{\varepsilon}$) and returns a $(1+\varepsilon)$-approximate
solution, for any desired $\varepsilon>0$. This result naturally complements
the problem's hardness for treewidth and is presented in \cref{sec:tw-approx}.
We complete that section by presenting a simple Win/Win argument which extends
our algorithm to an algorithm that is FPT parameterized by clique-width plus
$k$; this is based on a result of~\cite{wg/GurskiW00} stating
that graphs of bounded clique-width with no large complete bipartite subgraphs
actually have bounded treewidth.

\subsection{FPT Algorithm Parameterized by Treewidth and
Congestion}\label{sec:fpt-tw-k}

In this section, we prove the following theorem:

\begin{theorem}\label{thm:tw-alg}
Let $G$ be a graph with treewidth $w$, and let $k > 0$.
There is an algorithm that finds a spanning tree of $G$ with congestion $k$, if it exists, and runs in time ${(w+k)}^{\bO(w)} n^{\bO(1)}$.
\end{theorem}

Before we proceed with the formal proof of the theorem, we will give an intuition of the algorithm. %
Following the usual structure for bounded-treewidth graphs, we will use dynamic programming to compute solutions to the subgraphs corresponding to subtrees of the tree decomposition $\mathcal{T}$. %
For each such subgraph, we consider different possibilities for the solution to manifest on the root bag of the subtree; these different possibilities, which we call \emph{states}, represent equivalent classes of solutions that are ``compatible'' with respect to the rest of the graph, and thus it is sufficient to compute a feasible solution for each such class. %

To obtain the algorithm, we simply need to specify the states, as well as how to recursively obtain feasible solutions for each. %
Our states are composed of a tree, which we call the \emph{skeleton}, as well as values of congestion for each of its edges. %
The skeleton of a solution $T$ for a given bag $X_t$ intuitively represents how $T$ connects the vertices of $X_t$, %
and it can be obtained by eliminating vertices in $V \setminus X_t$ with degree 1 or 2 in $T$ (by contracting one of the incident edges to each of these vertices). %
Thus, the skeleton is a tree containing the vertices in $X_t$, plus at most $|X_t|$ vertices with degree at least $3$. %
The congestion values on an edge of the skeleton correspond to the congestion induced on that edge by all of the edges in $\mathcal T_t$, the subtree of $\mathcal T$ rooted at $t$.

There is one further particularity that we must consider: when constructing a skeleton from a solution, some of the resulting edges may represent paths in the subtree rooted at $t$, while others represent paths outside of this subtree. %
Thus, we label each edge of the skeleton with one of three types: a \emph{present} edge is simply an edge between two vertices in $X_t$; a \emph{past} edge represents a path using edges not in $E(X_t)$ but contained in $\mathcal T_t$; and a \emph{future} edge represents a path outside $\mathcal T_t$, that is, one that is not (yet) in the solution, but that must be added to make it compatible with the state. %
We similarly label the vertices with the type \emph{present} if they are in $X_t$, \emph{past} if they are not in $X_t$ but in $\mathcal T_t$, and \emph{future} otherwise.

Throughout the section, we assume familiarity with the definition and usual notation for treewidth (see e.g.~\citep[Chapter 7]{books/CyganFKLMPPS15}), including the notions of tree decomposition and nice tree decomposition. %
We assume that all tree decompositions are rooted: in the case of nice tree decompositions, the root is implicit in the definition; otherwise, any vertex can be chosen as the root.
When referring to subgraphs of $G$, we often refer to subgraphs induced by subtrees of $\mathcal T$: %
for $t\in V(\mathcal T)$, we denote by $\mathcal T_t$ the subtree of $\mathcal T$ rooted at $t$, %
and for $U \subseteq V(\mathcal T)$, we denote by $G[U]$ the subgraph of $G$ induced by the vertices contained in bags of $U$, i.e.\ the subgraph $G\bracks[\big]{\bigcup_{t \in U} X_t}$;
for convenience of notation, we write $G[\mathcal T_t]$ instead of $G[V(\mathcal T_t)]$. %


\paragraph*{Skeletons.}
In the course of the algorithm, we use a structure which we call \emph{skeleton}, which represents the relevant information of a spanning tree at a bag.

\begin{definition}\label{def:tw-alg:skeleton}
Given a graph $X$, a \emph{skeleton} $(S, \ell)$ for $X$ is a tree $S$ together with a labeling $\ell \colon E(S) \cup V(S) \to \{-1,0,1\}$, such that:
\begin{enumerate}
	\item $V(X) \subseteq V(S)$;
	\label{def:tw-alg:skeleton:vertices}
	\item for every $v \in V(S) \setminus V(X)$, $\deg_S(v) \geq 3$;
	\label{def:tw-alg:skeleton:deg-3}
	\item for every $v \in V(S)$, $\ell(v) = 0$ if and only if $v \in V(X)$;
	\label{def:tw-alg:skeleton:lbl-vtx-0}
	\item for every $uv \in E(S)$, $\ell(u), \ell(v) \in \braces{0,\ell(uv)}$.
	\label{def:tw-alg:skeleton:edge-labels}
	\item for every $uv \in E(S)$, if $\ell(uv) = 0$ then $uv \in E(X)$;
	\label{def:tw-alg:skeleton:lbl-edges-0}
\end{enumerate}

If $(S, \ell)$ satisfies every property except Property~\ref{def:tw-alg:skeleton:deg-3}, we call it a \emph{quasi-skeleton}.

We denote by $\ell^{-1}(i)$, $i \in \braces{-1,0,1}$ the graph obtained from the edges with label $i$ and their endpoints (which can have label $i$ or $0$).
\end{definition}

We use the labeling to represent the type of edges and vertices: $-1$ for past, $0$ for present and $1$ for future. %
The constraints in the definition ensure that present edges and vertices are contained in the bag, and that past edges are not incident on future vertices (and vice-versa).


Given a spanning tree $T$ and a node $t \in \mathcal{T}$, we can get a quasi-skeleton $(T,\ell)$ for $G[X_t]$ by appropriately labeling the vertices and edges of the spanning tree as follows:
\begin{itemize}
	\item in $G[X_t]$: for $x \in X_t$ and $x \in E(X_t)$ we set the label $\ell(x) = 0$;
	\item in $G[\mathcal T_t] - E(X_t)$: for $x \in V(G[\mathcal T_t]) \setminus X_t$ and $x \in E(G[\mathcal T_t]) \setminus E(X_t)$ we set $\ell(x)=-1$;
	\item in $G - G[\mathcal T_t]$: for $x \in V(G) \setminus V(G[\mathcal T_t]))$ and $x \in E(G) \setminus E(G[\mathcal T_t]))$ we set $\ell(x) = 1$.
\end{itemize}
The properties of quasi-skeletons follow for $(T,\ell)$ because $X_t$ is a separator.


\begin{figure}[tbh]
\centering
\includegraphics[scale=1,page=1]{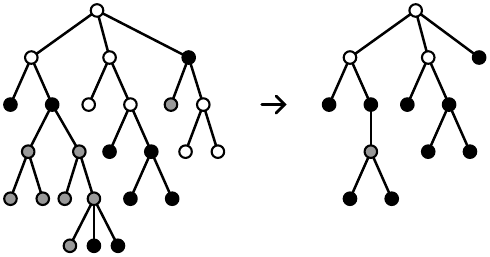}
\caption{
An example quasi-skeleton (left) and skeleton resulting from simplification (right).
Vertex labels are represented as colors: grey for $-1$ (past), black for $0$ (present) and white for $1$ (future);
edge labels are omitted and can be assumed to be $0$ when connecting two black vertices, and must be $-1$ (resp.\ $1$) when incident to a grey (resp.\ white) vertex.}
\label{fig:tw-alg:skeleton}
\end{figure}

To get a skeleton of $T$, we further need to transform the quasi-skeleton above into a skeleton by enforcing Property~\ref{def:tw-alg:skeleton:deg-3}.
For this purpose, we introduce the operation of \emph{simplification}, which eliminates vertices of degree 1 and 2, and thus converts a quasi-skeleton into a skeleton.
This operation is also used in the dynamic program, in order to modify a skeleton for a bag $X_t$ into a skeleton for its child bag.
\Cref{fig:tw-alg:skeleton} presents an example quasi-skeleton (tree with labels assigned) and the skeleton that results from applying simplification.

As we often use skeletons together with a function assigning a congestion value to each of its edges, the operation shows how to transform this function concurrently.

\begin{definition}[Simplification]
\label{def:tw-alg:simplification}
Let $X$ be a graph, $(S, \ell)$ be a quasi-skeleton for that graph, and $c\colon E(S) \to [k]$.

The \emph{simplification} of $(S, \ell)$ with function $c$ defines a new quasi-skeleton $(S',\ell')$ and function $c'\colon E(S') \to [k]$ as follows:
start with $(S', \ell') = (S, \ell)$, $c' = c$ and for each vertex $v \in V(S')\setminus V(X)$ with degree less than $3$, apply the following rules: %
\begin{itemize}
\item if $v$ has degree 1, we remove $v$ and its incident edge from $S'$, $\ell'$, $c$; %
\item if $v$ has degree 2 and label $\lambda = \ell(v)$, we replace $v$ and its incident edges by an edge $uw$ between its neighbors $u$ and $w$, set $\ell'(uw) = \lambda$, and adjust $c'(uw) = \max\{c(uv), c(vw)\}$. %
\end{itemize}
\end{definition}

\begin{lemma}\label{lem:tw-alg:quasi-skel}
The simplification of a quasi-skeleton is a skeleton.
\end{lemma}

\begin{proof}
When the process completes, $(S', \ell')$ must satisfy Property~\ref{def:tw-alg:skeleton:deg-3}, as we iteratively remove vertices that do not satisfy it. %

For the remaining properties, we will show that an application of the lemma for a single step preserves them, and thus by repeated application the lemma holds. %
Properties~\ref{def:tw-alg:skeleton:vertices} and~\ref{def:tw-alg:skeleton:lbl-vtx-0} are not affected by removing vertices, and if an edge $uw$ is created, it has label $\ell(uw) = \ell(v) \in \braces{-1,1}$, thus Property~\ref{def:tw-alg:skeleton:lbl-edges-0} is satisfied. %
Finally, since $\ell(v) \neq 0$, $\ell(uv) = \ell(vw) = \ell(v)$, and thus Property~\ref{def:tw-alg:skeleton:edge-labels} is satisfied by setting $\ell(uw) = \ell(v)$.
\end{proof}

\paragraph*{Algorithm.}
We define a dynamic programming table $D$ with entries for every node $t \in \mathcal T$ and every triple $(S, \ell, c)$, where $(S, \ell)$ is a skeleton for $G[X_t]$ and $c\colon E(S) \to [0,k]$ is a function representing an upper bound on the congestion of each edge in $S$. %

To compute a solution, the algorithm simply uses dynamic programming:
it starts by computing a nice tree decomposition $(\mathcal T, \mathcal X)$ for $G$ with width at most $2w+1$ and at most $\bO(n w)$ bags, which can be computed in time $2^{\bO(w)} \cdot n$~\citep{Korhonen21}; %
then it computes every entry of $D$, and then outputs the single entry for the root $r$ of $\mathcal T$, the tree $T = D[r, (\varnothing, \varnothing, \varnothing)]$.

Hence, we only need to define the entries of $D$, as well as recursion rules to compute each entry.
Informally, $D[t, (S, \ell, c)]$ represents a forest $F_t \subseteq G[\mathcal T_t]$ such that the congestion induced by the edges of $G[\mathcal T_t]$ is at most $k$ for the edges of $F_t$, and connecting the trees of $F_t$ according to the skeleton leads to congestion bounded by the function $c$ on each edge of the skeleton.
In particular, the edges of $\ell^{-1}(1)$ (future edges) represent how the different trees of $F_t$ are connected in the final solution, and for $e \in \ell^{-1}(1)$, $c(e)$ represents the congestion induced by edges of $G[\mathcal T_t]$ on a path that is represented by skeleton edge $e$ and which must be present in the final solution.

Using \Cref{fig:tw-alg:skeleton} as an example, let $T$ be the tree on the left, $X_t$ be the set of black vertices (label $0$), and $(S,\ell)$ be the skeleton on the right, $F_t$ would be the forest obtained by removing the white vertices (label $1$) from $T$.
The different trees of $F_t$ are connected according to the edges of $\ell^{-1}(1)$, thus if we add to $F_t$ the white vertices in the skeleton, as well as their incident edges, we get a tree that contains all of the past and present edges and vertices of the solution, but a simplification of the future subgraph, that is, we get a tree $T'$ obtained from $T$ by applying simplification only to the white vertices.

We formalize the desired properties of entries of $D$ by the definition below:

\begin{definition}\label{def:tw-alg:props}
Let $t \in \mathcal T$, $(S, \ell)$ be a skeleton for $G[X_t]$ and $c\colon E(S) \to [0,k]$, and assume $V(S) \cap V(G) = X_t$.

We say that a forest $F_t\subseteq G[\mathcal{T}_t]$ is a \emph{consistent solution} for $(t, (S, \ell, c))$ if:
\begin{enumerate}
	\item $F_t$ is a forest of $G[\mathcal T_t]$ on the same vertex set, that is, $V(F_t) = V(G[\mathcal{T}_t])$;
	\label{def:tw-alg:props:subgraph}
	\item For $uv \in E(S)$, if $\ell(uv) = 0$, then $uv \in F_t$;
	\label{def:tw-alg:props:label0}
	\item For $uv \in E(S)$, if $\ell(uv) = -1$, then $F_t$ contains a $u$-$v$-path $F_{uv}$ with edges from  ${G[\mathcal T_t] - E(X_t)}$;
	\label{def:tw-alg:props:path-1}
	\item $T' := F_t \cup \ell^{-1}(1)$ is a tree;
	\label{def:tw-alg:props:tree}
	\item For every edge $e \in E(S)$ with $\ell(e) \in \braces{0,1}$, the
	congestion in $T'$ induced by $G[\mathcal T_t] - E(X_t)$ on $e$ is $c(e)$;
	\label{def:tw-alg:props:cong-fut}
	\item For every edge $e \in E(S)$ with $\ell(e) = -1$, the
	congestion in $T'$ induced by $G[\mathcal T_t] - E(X_t)$ on every edge of $F_e$ is at most $c(e)$;
	\label{def:tw-alg:props:cong-past}
	\item For every $e \in E(T')$, the congestion in $T'$ induced by $G[\mathcal T_t] - E(X_t)$ on $e$ is at most $k$.
	\label{def:tw-alg:props:cong-ext}
\end{enumerate}
\end{definition}

We will now see how to construct the entries $D[t, (S, \ell, c)]$ recursively.
An entry can be set as \emph{invalid}, in which case it does not have a value; these entries either do not correspond to solutions or have congestion above $k$.
For each rule, we present a figure showing how the skeletons in the child bag(s) relate to the skeletons for $X_t$.
\begin{itemize}
	\item \textbf{Leaf $t$:} $D[t, (S, \ell, c)] = (\varnothing, \varnothing)$ for the only possible skeleton $(S, \ell)$ and function $c$, where $S=(\varnothing, \varnothing)$ and $\ell, c$ are empty functions.
	\item \textbf{Forget node $t$:} let $t'$ be the child of $t$ and let $\braces{v} = X_{t'} \setminus X_t$.
	For any solution $D[t', (S', \ell', c')]$, if $v$ has an incident edge with label $1$, no solution for $t$ exists for $(S', \ell', c')$; otherwise, we set $D[t, (S, \ell, c)]=D[t', (S', \ell', c')]$ for $(S, \ell, c)$ obtained as follows:
	\begin{enumerate}
		\item we start with $(S, \ell, c) = (S', \ell', c')$;
		\item we add the congestion corresponding to the forgotten vertex $v$: for each edge $uv \in E(G)$, $u \in X_t$, increment $c(e)$ on each edge $e$ on the $u$-$v$-path in $S$; if $c(e) > k$ for any edge, the process stops and we disregard $(S', \ell', c')$;
		\item we mark $v$ as a past vertex ($\ell(v) = -1$) and apply simplification of $(S,\ell)$ with function $c$, but for the new graph $G[X_t]$.
	\end{enumerate}
	An example for $(S', \ell')$ and resulting $(S, \ell)$ is presented in \Cref{fig:tw-alg:op-forget}

	\item \textbf{Introduce node $t$:} 
	let $t'$ be the child of $t$ and let $\braces{v} = X_{t} \setminus X_{t'}$. %
	To avoid case enumeration, we describe the reverse process, i.e.~how to uniquely get $(S', \ell')$ from $(S,\ell)$.
	\Cref{fig:tw-alg:op-introduce} presents the different structures of $(S,\ell)$ for a given example $(S', \ell')$; we describe the four cases in \cref{lem:tw-alg:correctn}.

	Given $(S, \ell, c)$, the single entry $(S', \ell', c')$ used to compute $D[t,(S, \ell, c)]$ is obtained as follows:
	start by setting $(S', \ell', c') = (S, \ell, c)$, then set $\ell'(v) = 1$, and $\ell'(uv)$ for every $uv \in E(X_t)$ and apply simplification. %
	If $D[t', (S', \ell', c')]$ exists, we construct $D[t, (S, \ell, c)]$ from $F_{t'}=D[t', (S', \ell', c')]$ by adding to $F_{t'}$ the vertex $v$ and its incident edges $e \in E(S)$ that have label $\ell(e)=0$.

	\item \textbf{Join node $t$:} given solutions $D[t_1, (S, \ell_1, c_1)]$, $D[t_2, (S, \ell_2, c_2)]$, a valid solution can be obtained if all of the following hold:
	\begin{enumerate}
		\item $c_1(e) + c_2(e) \leq k$, for $e \in E(S)$;
		\item $\ell_1(x) = 0$ if and only if $\ell_2(x) = 0$, for $x \in V(S) \cup E(S)$;
		\item $\ell_1(x)$ and $\ell_2(x)$ are not both $-1$, for $x \in V(S) \cup E(S)$.
	\end{enumerate}

	We set $D[t, (S, \ell, c)] = D[t_1, (S, \ell_1, c_1)] \cup D[t_2, (S, \ell_2, c_2)]$ for $\ell(x) = \min\{\ell_1(x), \ell_2(x)\}$, ${x \in V(S) \cup E(S)}$ and $c(e) = c_1(e) + c_2(e)$, $e \in E(S)$.
	An example is presented in \Cref{fig:tw-alg:op-join}
\end{itemize}

\begin{figure}[t]
\centering
\includegraphics[scale=1,page=1]{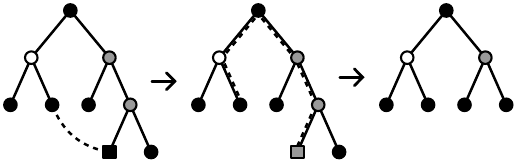}
\caption{
An example skeleton $(S', \ell')$ (left) with vertex $v$ to forget (square), resulting quasi-skeleton after modifying the label of $v$ (middle), and skeleton $(S,\ell)$ (right).
For each edge $uv \in E(X_t)$ (dashed edge, left), we add congestion along the $u$-$v$-path in the tree (dashed path, middle). \\
Vertex labels represented as grey for $-1$ (past), black for $0$ (present) and white for $1$ (future).}
\label{fig:tw-alg:op-forget}
\end{figure}

\begin{figure}[t]
\centering
\includegraphics[scale=1,page=2]{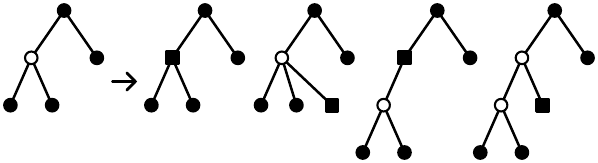}
\caption{
An example skeleton $(S', \ell')$ (left), and different possibilities for skeleton $(S,\ell)$ (right) after inserting a vertex $v$ (square).
These four cases are described in the proof of \cref{lem:tw-alg:correctn}, and described in the dynamic program by the reverse, i.e.~how to get $(S', \ell')$ from $(S,\ell)$; every skeleton $(S,\ell)$ on the right generates the same $(S', \ell')$ on the left. \\
Vertex labels represented as grey for $-1$ (past), black for $0$ (present) and white for $1$ (future).}
\label{fig:tw-alg:op-introduce}
\end{figure}

\begin{figure}[t]
\centering
\includegraphics[scale=1,page=3]{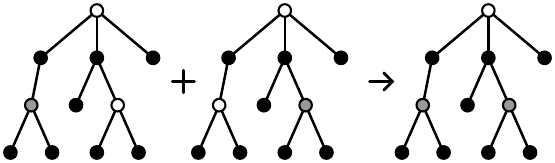}
\caption{
Example skeletons $(S, \ell_1)$ (left) and $(S, \ell_2)$ (middle), and the resulting skeleton $(S, \ell)$ (right).
Skeletons $(S, \ell_1)$ and $(S, \ell_2)$ must be the same, except for past vertices, which are disjoint. \\
Vertex labels represented as grey for $-1$ (past), black for $0$ (present) and white for $1$ (future).}
\label{fig:tw-alg:op-join}
\end{figure}


To obtain a solution to the problem, we simply compute $T = D[r, (\varnothing, \varnothing, \varnothing)]$. %
If $T$ is a consistent solution, then it is a forest containing the vertices $V(G)$ by Property~\ref{def:tw-alg:props:subgraph}, it is a tree by Property~\ref{def:tw-alg:props:tree}, and has congestion at most $k$ by Property~\ref{def:tw-alg:props:cong-ext} (since $X_r = \varnothing$). %
We therefore conclude that $T$ is a spanning tree of $G$ with congestion at most $k$, as desired. %

The following lemmas complete the proof:
\Cref{lem:tw-alg:correctn} shows that every $D[t, (S, \ell, c)]$ is a consistent solution if it is not invalid,
and \cref{lem:tw-alg:completen} shows that for any tree $T$ with congestion at most $k$, there is a corresponding entry for each node $t$ (that is not invalid), thus $T$ is a consistent solution; finally,
\cref{lem:tw-alg:runtime} shows that the algorithm runs in the stated running time.

\begin{lemma}\label{lem:tw-alg:correctn}
For any $t$ and $(S, \ell, c)$, if $D[t, (S, \ell, c)]$ exists, then it is a consistent solution.
\end{lemma}

\begin{proof}
	If $t \in V(\mathcal T)$ is a leaf, then the properties of
	Definition \ref{def:tw-alg:props} hold trivially, as $G[\mathcal T_t]$ is an empty graph.
	We will now show by induction (from the leaves up) that the claim holds for every node $t$. %
	Let $t \in V(\mathcal T)$ be any node, and $(S, \ell, c)$ such that $F_t := D[t, (S, \ell, c)]$ exists. %

	\begin{itemize}
	\item \textbf{If $t$ is an introduce node,}
	then $F_t = D[t', (S', \ell', c')] \cup \ell^{-1}(0)$ for some $(S', \ell', c')$. %
	In the dynamic program, $(S', \ell', c')$ is obtained from $(S, \ell, c)$ by setting $\ell'(v) = 1$ and applying simplification. %
	In this case, simplification means that if $v$ has degree $2$, it is eliminated by connecting its two neighbors; if $v$ has degree $1$, it is directly removed, and its parent is eliminated if it now has degree $2$. %

	If we consider the process in reverse, we see how to obtain $(S, \ell, c)$ from $(S', \ell', c')$,
	which corresponds to one of these cases:
	\begin{enumerate}[i.]
		\item take a vertex of $S'$ with label $1$ and identify it with $v$;
		\item take an edge of $S'$ with label $1$ and subdivide it, with $v$ being the new vertex;
		\item take a vertex of $S'$ with label $0$ or $1$ and add $v$ as its child;
		\item take an edge of $S'$ with label $1$, subdivide it, and add $v$ as the child of the new vertex; the new vertex is assigned label $1$.
	\end{enumerate}
	In either of those cases, we assign to the new edges the label $0$ if both of their endpoints have label $0$, and the label $1$ otherwise.
	In the case of subdivision, we assign the same congestion to the new edges as the subdivided edge; if $v$ is added as a child, its incident edge is assigned edge $0$.
	The labels and congestion are assigned so that the properties of Definitions \ref{def:tw-alg:props} are satisfied.

	Furthermore, $G[\mathcal T_t] - E(X_t)$ is equal to $G[\mathcal T_{t'}] - E(X_{t'})$ with an additional isolated vertex $v$, since all of the edges incident on $v$ in $G[\mathcal T_t]$ are also in $E(X_t)$. %
	Thus, as we only count congestion of edges in $G[\mathcal T_t] - E(X_t)$, congestion does not change.

	\item \textbf{If $t$ is a forget node,} we first show that increasing the congestion does not change the properties.
	For every new edge $uv \in E(X_{t'})$ incident on $v$, the construction increments $c(e)$ for every edge $e$ on the $u$-$v$-path in $S$;
	this precisely matches the edges of $T'$ whose congestion is increased by the edge $uv$, as edges with label $0$ and $1$ are present in $T'$, and edges with label $-1$ are present in the form of a path.
	Property~\ref{def:tw-alg:props:cong-ext} is preserved since we check that no edge exceeds the congestion value of $k$.

	The operations used in simplification do not change $F_t$ and do not add any edges with label $0$ or $1$, but affect Properties~\ref{def:tw-alg:props:path-1},~\ref{def:tw-alg:props:cong-past}, in particular when handling a vertex $v'$ with degree 2.
	However, the action of eliminating $v'$ simply concatenates the paths $P_{uv'}$ and $P_{v'w}$ in $F$ to obtain the path $P_{uw}$, and setting $c(uw) = \max\{c(uv'), c(v'w)\}$ ensures that each edge in $P_{uw}$ has capacity at most $c(uw)$.

	\item \textbf{If $t$ is a join node,} $F_t$ is obtained as the union of solutions $F_1 := D[t_1, (S, \ell_1, c_1)]$, $F_2 := D[t_2, (S, \ell_2, c_2)]$ for the children nodes $t_1$, $t_2$, and we assume that for every $e \in E(S)$, $c_1(e) + c_2(e) \leq k$ and for every $x \in V(S) \cup E(S)$, $\ell_1(x) = 0$ if and only if $\ell_2(x) = 0$, and $\ell_1(x)$ and $\ell_2(x)$ are not both $-1$. %

	Properties~\ref{def:tw-alg:props:subgraph},~\ref{def:tw-alg:props:label0},~\ref{def:tw-alg:props:path-1} follow by definition of $F_t$ and of $\ell$ as $\ell(x) = \min\{\ell_1(x), \ell_2(x)\}$. %
	In particular, the definition of $\ell$ implies that for every edge $uw \in E(S)$ such that $\ell(uw) = -1$, either $\ell_1(uw) = -1$ or $\ell_2(uw) = -1$, and thus Property~\ref{def:tw-alg:props:path-1} implies that a $u$-$w$-path exists in either $F_1$ or $F_2$, and thus in $F_t$. %
	We can similarly see that $T = F_t \cup \ell^{-1}(1)$ is a tree, as we are equivalently taking $S$ and replacing each $uw \in E(S)$ with label $-1$ by a subtree that connects $u$ and $w$; as we know that $F_1 \cup \ell_1^{-1}(1)$ and $F_2 \cup \ell_2^{-1}(1)$ are both trees by induction, and that $\ell_1^{-1}$ and $\ell_2^{-1}$ are edge-disjoint (as they use edges from disjoint subtrees of $\mathcal T$), then $T$ must be acyclic, and thus a tree. %
	Finally, we can see that each edge in $G[\mathcal T_t] - E(X_t)$ must belong to either $G[\mathcal T_{t_1}] - E(X_t)$ or $G[\mathcal T_{t_2}] - E(X_t)$, and thus the corresponding congestion is accounted for in $S_1$ or $S_2$, respectively;
	by summing these congestion values, we ensure that Properties~\ref{def:tw-alg:props:cong-fut},~\ref{def:tw-alg:props:cong-past},~\ref{def:tw-alg:props:cong-ext} are satisfied. %
	\end{itemize}
\end{proof}

\begin{lemma}\label{lem:tw-alg:completen}
Let $T$ be a tree with congestion at most $k$.
Then for any $t$, there is $(S_{t}, \ell_{t}, c_{t})$ such that $D[t, (S, \ell, c)]$ exists, i.e.\ is not invalid.
\end{lemma}

\begin{proof}
	We start by defining, for each $t \in V(\mathcal T)$, a solution $F_t$ and a state $(S_t, \ell_t, c_t)$, and then show that $F_t$ is consistent solution for $D[t, (S_t, \ell_t, c_t)]$, and thus $D[t, (S_t, \ell_t, c_t)]$ exists.

	The solution $F_t$ can be obtained by restricting $T$ to the corresponding subgraph, that is, $F_t := T \cap G[\mathcal T_t]$. %
	$S_t$ can be obtained by labeling the vertices of $T$ according to $\mathcal T$, and then applying simplification. %
	Concretely, we set $\ell(v) = 0$, $v \in X_t$ and $\ell(e) = 0$, $e \in T \cap E(X_t)$;
	then, we set $\ell(v) = -1$ for $v \in V(G[\mathcal T_t]) \setminus X_t$ and $\ell(v) = 1$ if $v \in V \setminus V(G[\mathcal T_t])$; finally, an edge $e \in T \setminus E(X_t)$ gets the non-zero label of one of its endpoints. %
	For every edge $e \in E(S_t)$, we take $c(e)$ to be the congestion induced by the edges in $G'_t := {G[\mathcal T_t] - E(X_t)}$ on $e$.
	We can now see that $(T, \ell_t)$ is a quasi-skeleton for $G[X_t]$, and thus, by \cref{lem:tw-alg:quasi-skel}, we obtain the skeleton $(S_t, \ell_t)$ by simplification.

	Next we show that $F_t$ is a consistent solution: %
	by definition, it is a forest containing every vertex in $G'_t$, every edge with label $0$, and a path connecting the endpoints of every edge with label $-1$ (as such edges are obtained by simplification of vertices with label $-1$); %
	similarly, we can see that $T'$ is a tree, as it can be constructed from $T$ by simplification of vertices with label $1$; %
	the final properties can be proved by seeing that $c_t(uw)$ corresponds to the maximum congestion in the $u$-$w$-path in $T$, and that $T$ has congestion at most $k$.

	As $(S_t, \ell_t, c_t)$ can be constructed from $(S_{t'}, \ell_{t'}, c_{t'})$ (for introduce and forget nodes) or from $(S_{t_1}, \ell_{t_1}, c_{t_1})$ and $(S_{t_2}, \ell_{t_2}, c_{t_2})$ (for join nodes) using the recursion rules, then $D[t, (S_t, \ell_t, c_t)]$ can be obtained from the solutions of the children nodes, and by induction on $t$, $D[t, (S_t, \ell_t, c_t)]$ exists.
\end{proof}

\begin{lemma}\label{lem:tw-alg:runtime}
The running time of the algorithm is ${(w+k)}^{\bO(w)} n^{\bO(1)}$.
\end{lemma}

\begin{proof}
Let $t \in V(\mathcal T)$. %
We start by showing that the number of possible states $(S, \ell, c)$ is ${(w+k)}^{\bO(w)}$. %
If $(S, \ell)$ is a skeleton, $S$ can have at most $w+1$ vertices with degree at most $2$, since they are in $X_t$, and thus it can have at most $w$ vertices with degree at least $3$, for a total of $2w+1$. %
As a consequence, we can specify $(S, \ell, c)$ by specifying, for each of the at most $2w+1$ vertices in $S$, which vertex is its parent in $S$ ($2w+1$ choices) and its label, and for each edge $e$, its label and value $c(e) \in [k]$. %
In total, the number of possibilities is upper-bounded by ${(w+k)}^{\bO(w)}$.

The running time of the algorithm is determined by the complexity of recursively calculating each of the ${(w+k)}^{\bO(w)}$ values for the $\bO(n w)$ bags in $(\mathcal T, \mathcal X)$. %
In the case of forget nodes, we can compute $(S, \ell, c)$ from $(S', \ell', c')$ in polynomial time, by updating the congestion values and then applying simplification to the skeleton. %
Similarly, for introduce nodes, we can compute $(S', \ell', c')$ from $(S, \ell, c)$ in polynomial time, by applying simplification.
For join nodes, given states for $t_1$, $t_2$, we simply have to compute the new labeling and congestion values; we can enumerate the states for $t_1$ and $t_2$ and compute  $(S, \ell, c)$ in time ${(w+k)}^{\bO(w)}$. %

Overall, we can compute the solutions for a node $t$ in time ${(w+k)}^{\bO(w)}$, and thus the algorithm runs in time ${(w+k)}^{\bO(w)} \cdot n^{\bO(1)}$.
\end{proof}

\subsection{FPT Approximation and Clique-width}\label{sec:tw-approx}

In this section we leverage the algorithm of \cref{thm:tw-alg} to obtain two
further results. On the one hand, we observe that the FPT algorithm of
\cref{thm:tw-alg} relies on two parameters (treewidth and the optimal
congestion), but it is likely impossible to improve this to an exact algorithm
parameterized by treewidth alone (indeed, we saw that \cref{thm:disjoint-union}
implies that the problem is hard even for parameters much more restricted than
treewidth). We are therefore motivated to consider the question of
approximation and present an FPT approximation scheme parameterized by
treewidth alone. Our algorithm is based on a combination of the algorithm of
\cref{thm:tw-alg} together with a technique introduced by~\cite{icalp/Lampis14}
(later also used among others in~\citep{dmtcs/AngelBEL18,BelmonteLM20,ChuL23,dmtcs/Lampis26})
which allows us to perform dynamic programming while maintaining approximate values for the
congestion. The end result, presented in \cref{thm:tw-apx}, is an FPT approximation scheme which,
for any $\varepsilon>0$, returns a $(1+\varepsilon)$-approximate solution in
time ${(w/\varepsilon)}^{\bO(w)}n^{\bO(1)}$, that is, FPT in
$w+\frac{1}{\varepsilon}$.

Having established this, we obtain a second extension of \cref{thm:tw-alg}, to
the more general parameter clique-width. Here, we rely on a Win/Win argument:
suppose we are given an input graph $G$ of clique-width $w$ and are asked if a
tree of congestion $k$ can be found; if $G$ contains a large bi-clique (in
terms of $k$), then we show that this can be found and we can immediately say
No; otherwise, by a well-known result of~\cite{wg/GurskiW00},
we infer that the graph actually has low treewidth, so we can apply
\cref{thm:tw-alg}.
The result is formalized in \cref{thm:cw}.

\begin{theorem}\label{thm:tw-apx} There is an algorithm which, for all
$\varepsilon>0$, when given as input a graph $G$ of treewidth $w$, returns a
spanning tree of congestion at most $(1+\varepsilon) \stc(G)$ in time
${(\frac{w}{\varepsilon})}^{\bO(w)}n^{\bO(1)}$.
\end{theorem}

\begin{proof}
We will modify the algorithm of \cref{thm:tw-alg}. As before, we assume that we
start with a nice tree decomposition of width $\bO(w)$, but now we further apply a
result of~\cite{siamcomp/BodlaenderH98} which allows us
to edit the nice tree decomposition so that it now has height (maximum distance from
the root to any leaf) $\bO(w^2\log n)$, while the width remains
$\bO(w)$.  It will allow us to bound the accumulated errors of the dynamic programming algorithm. Let $h$ be the height of the resulting nice tree
decomposition.

We now modify the dynamic programming table of the algorithm of
\cref{thm:tw-alg} as follows:  instead of triples $(S,\ell,c)$, we will store
triples $(S,\ell,\hat{c})$, where $\hat{c}$ is now a function on $E(S)$ with
range $\{0\} \cup \setdef{{(1+\delta)}^i}{i \in \mathbb{N}, \, {(1+\delta)}^i \le
(1+\varepsilon)k}$. In other words, instead of storing for each edge of $S$
its current congestion (which is an integer in $[0,k]$), we store an integer
power of $(1+\delta)$, where $\delta$ is a parameter we define as $\delta =
\frac{\varepsilon}{2h}$.

In order to modify the algorithm we revisit all the places where a value $c(e)$
is calculated, based on the graph or previously calculated values. This happens
in:

\begin{itemize}

\item Forget nodes: here we are constructing a triple $(S,\ell,\hat{c})$ from a
triple $(S',\ell',\hat{c}')$ as we are forgetting vertex $v$. Previously, for
each $uv\in E(G)$, with $u\in X_t$, we were incrementing $c(e)$ on each edge
$e$ used in the $u$-$v$-path in $S$. Suppose that for $e\in E(S)$ this process
sets $c'(e)=c(e)+r_e$. We set $\hat{c}'(e)$ to be equal to the smallest integer
power of $(1+\delta)$ that is at least as large as $\hat{c}(e)+r_e$ (in other
words, we add $r_e$ to $\hat{c}(e)$ and round up to the nearest integer power
of $(1+\delta)$). If for any edge $\hat{c}(e)>(1+\varepsilon)k$ we discard the
solution as invalid. Other operations are unchanged, except we use $\hat{c}$ in
the place of $c$.

\item Introduce nodes: Operations are unchanged, except we use $\hat{c}$ in the
place of $c$.

\item Join nodes: previously we were calculating a triple $(S,\ell,c)$ from two
triples $(S,\ell_1,c_1)$ and $(S,\ell_2,c_2)$, rejecting solutions for which for
some edge $e$ we had $c_1(e)+c_2(e)>k$. Now we reject solutions if
$\hat{c}_1(e)+\hat{c}_2(e)>(1+\varepsilon)k$. If a combination is not rejected,
we were previously calculating $c(e)=c_1(e)+c_2(e)$. Now we set $\hat{c}(e)$ to
be the smallest integer power of $(1+\delta)$ that is at least as large as
$\hat{c}_1(e)+\hat{c}_2(e)$, that is, we perform the same addition as before
and round-up to the closest integer power of $(1+\delta)$. Other operations are
unchanged.

\end{itemize}

Let us now argue that the algorithm indeed produces a
$(1+\varepsilon)$-approximate solution. We initially guess $\stc(G)$ and assume
we execute the (modified) algorithm of \cref{thm:tw-alg} for $k=\stc(G)$. We
define the height of a bag of the decomposition as the largest distance from the
bag to a leaf rooted in its sub-tree. Leaves are at height $0$ and the root is
at height $h$. Before we proceed, we note that we have set $\delta$
to a value such that ${(1+\delta)}^h \le e^{\delta h} = e^{\frac{\varepsilon}{2}}
\le (1+\varepsilon)$ for small enough $\varepsilon$.

We want to maintain the following invariant:

\begin{itemize}

\item For each $t$ and triple $(S,\ell,c)$ such that the exact algorithm stores
a forest in $D[t, (S,\ell,c)]$, there exists $\hat{c}$ such that the new
algorithm stores a forest in $D[t, (S,\ell,\hat{c})]$ and for all $e\in
E(S)$ we have that $\hat{c}(e) \le {(1+\delta)}^{h_t} c(e)$, where $h_t$ is the height of
$t$.

\item For each $t$ and triple $(S,\ell,\hat{c})$ such that the new algorithm
stores a forest in $D[t, (S,\ell,\hat{c})]$, there exists $c$ such that the
exact algorithm stores a forest in $D[t, (S,\ell,c)]$ and for all $e\in E(S)$
we have that $c(e) \le \lceil\hat{c}(e)\rceil$.

\end{itemize}

Informally, we are claiming here that the approximation algorithm is slightly
over-estimating the congestion of every edge: on the one hand, if the
approximate algorithm computes some solution, an exact solution with at most
the same congestion exists; while on the other hand, if an exact solution
exists, the algorithm will compute a congestion upper-bounded by the correct
value multiplied by ${(1+\delta)}^{h_t}$.

It is now not hard to establish this invariant by induction:

\begin{itemize}

\item For a Forget node $t$ with child $t'$: we have $h_{t'} = h_t-1$. By
induction, for an exact triple $(S',\ell',c')$ we have calculated an
approximate triple $(S',\ell',\hat{c}')$ with $\hat{c}'(e)\le
{(1+\delta)}^{h_t-1} c'(e)$. For each edge $e$ the exact triple $(S,\ell,c)$
would contain a value $c(e)=c'(e)+r_e$, while we set $\hat{c}(e)\le (1+\delta)(
\hat{c}'(e)+r_e) \le {(1+\delta)}^{h_t}c(e)$, as desired. For the other
direction, we observe that we set $\hat{c}(e) \ge \hat{c}'(e)+r_e$, so using
that by induction $\hat{c'}(e)\ge c'(e)$ we get the second statement of the
invariant. We observe that discarding a solution is correct, because for all
valid triples calculated by the exact algorithm, for each value $c(e)\le k$, an
approximate triple will have value at most $\hat{c}(e)\le {(1+\delta)}^h c(e) \le
(1+\varepsilon)k$, hence solutions with $\hat{c}(e)>(1+\varepsilon)k$ are safe
to discard.

\item For Introduce nodes, congestion values are not modified, but the height
increases, so the invariant remains true.

\item For Join nodes, where a node $t$ has two children $t_1,t_2$, we have
$h_{t_1}, h_{t_2}\le h_t-1$. We use similar reasoning as in the Forget case, in
particular, we compute $\hat{c}(e) \le (1+\delta)(\hat{c}_1(e)+\hat{c}_2(e))$,
but by inductive hypothesis $\hat{c}_1(e)\le {(1+\delta)}^{h_t-1}c_1(e)$ (and
similarly for $\hat{c}_2$), so $\hat{c}(e)$ over-estimates the correct value by
at most a factor ${(1+\delta)}^{h_t}$. Similarly, $\hat{c}(e) \ge
\hat{c}_1(e)+\hat{c}_2(e)$, so if $\hat{c}_1(e)\ge c_1(e)$ (and similarly for
$\hat{c}_2$) we get the second statement. Again, discarding a solution with
congestion strictly greater than $(1+\varepsilon)k$ is justified as in the
Forget case.

\end{itemize}

It is not hard to see that the invariant guarantees that if a solution exists,
the approximation algorithm will find it, and that if the approximation
algorithm finds a solution of congestion estimated at $\hat{k}\le
(1+\varepsilon)k$, a solution of this congestion can be constructed. Let us
then bound the running time, which is dominated by the size of the DP tables.
As before, there are $w^{\bO(w)}$ trees $S$, $2^{\bO(w)}$ functions $\ell$, while
the number of possible values stored for $\hat{c}(e)$ is at most
$\bO(\log_{(1+\delta)}k) = \bO(\frac{\log k}{\log(1+\delta)}) = \bO(\frac{\log
n}{\delta}) = \bO(\frac{w^2\log^2 n}{\varepsilon})$, where we used that $h= \bO(w^2\log
n)$ and approximated $\log(1+\delta)$ by $\delta$. The total number of possible
DP entries is then ${(\frac{w\log n}{\varepsilon})}^{\bO(w)}$. We now use a standard
Win/Win argument: if $w\le \sqrt{\log n}$, then the running time bound we
calculated is $n^{o(1)}$ and the whole algorithm runs in polynomial time;
otherwise, $\log n< w^2$, so the running time of the algorithm is
${(\frac{w}{\varepsilon})}^{\bO(w)} n^{\bO(1)}$.
\end{proof}

\begin{theorem}\label{thm:cw} There is an algorithm which, when given as input
a graph $G$, an integer $k$, and a clique-width expression for $G$ with $w$
labels, correctly decides if $\stc(G)\le k$ in time ${(wk)}^{\bO(wk)}n^{\bO(1)}$.
\end{theorem}

\begin{proof}
We use the following claim:

\begin{claim} If, for some $t$, a graph $G$ has $K_{t,t}$ as a subgraph,
then $\stc(G)\ge t$.
\end{claim}

\begin{claimproof}
Theorem~1(a) of~\cite{dm/Ostrovskii04} states that ``$\stc(G) \ge m_{G}$'',
where $m_{G}$ is the maximum number of edge-disjoint paths between two vertices in $G$.
We can see that $m_{K_{t,t}} = t$. Hence, if a graph $G$ contains $K_{t,t}$ as a subgraph, $\stc(G)\ge t$.
\end{claimproof}

The algorithm now relies on a result of~\cite{wg/GurskiW00} which states that
if a graph has clique-width $w$ and does not contain $K_{t,t}$ as a subgraph,
then $G$ has treewidth at most $3tw$. Given $G$, $k$, and a clique-width
expression with $w$ labels, our algorithm first checks if $G$ has treewidth at
most $3(k+1)w$ by running the 2-approximation by~\cite{Korhonen21}. If this
algorithm does not return a tree decomposition of width at most $6(k+1)w+1$,
we can immediately conclude that the treewidth of $G$ is more than $3(k+1)w$.
However, if $G$ did not contain any bicliques $K_{k+1,k+1}$, this could not
happen, so $G$ must contain such a biclique, which by the claim allows us to
immediately return No. We are therefore left with the case where we have a tree
decomposition of width $\bO(kw)$, on which we execute the algorithm of \cref{thm:tw-alg}.
\end{proof}

\section{FPT Algorithms}\label{sec:fpt_algorithms}

In \cref{sec:treewidth}, we complemented the hardness results of
\cref{sec:hardness} with algorithms for treewidth plus congestion, an
FPT approximation scheme for treewidth, and an FPT algorithm for clique-width
plus congestion. We now turn to the question of which structural parameters by
themselves render {\STC} fixed-parameter tractable. The hardness results of
\cref{thm:disjoint-union,thm:hardness:modularwidth} rule out this possibility
for many natural candidates, so we focus on parameters that evade these
reductions. We present such results for distance to clique in
\cref{ssec:dist_to_clique}, vertex integrity in \cref{ssec:vi}, and feedback
edge number in \cref{ssec:fes}.

\subsection{Distance to Clique}\label{ssec:dist_to_clique}

\cref{cor:twin-cover} implies that {\STC} remains W[1]-hard even on very structured dense instances.
In this subsection we search for parameters that render the problem tractable on dense instances,
and present an FPT algorithm when parameterized by the distance to clique of the input graph,
arguably one of the most restrictive such parameters.
Interestingly, the running time of our algorithm is dictated by the ``easy'' case,
where the clique modulator is large with respect to the size of the graph.
We remark that a modulator to clique of a graph $G$ is a vertex cover in the complement $\bar{G}$ of $G$,
and thus the minimum modulator can be computed by employing any FPT algorithm for \textsc{Vertex Cover} (e.g.~\citep{tcs/ChenKX10,stacs/HarrisN24}) on $\bar G$.

\begin{theorem}\label{thm:distance_to_clique}
    Given $G = (V,E)$ and $S \subseteq V$ with $G-S$ being a clique,
    there is an algorithm that returns a spanning tree of $G$ of congestion $\stc(G)$ in time $2^{\bO(k^3)} n^{\bO(1)}$,
    where $n = |V|$ and $k = |S|$.
\end{theorem}

\begin{proof}
    Let $C = V \setminus S$ denote the vertex set of the clique $G-S$, with $c=|C|$.
    First consider the case where $c \le 2k^3 + 4k = \bO(k^3)$.
    Then, it holds that $n = c + k = \bO(k^3)$,
    and one can find a spanning tree of $G$ of congestion $\stc(G)$ in time $2^{\bO(k^3)} n^{\bO(1)}$ by using the $2^{n} n^{\bO(1)}$
    algorithm by~\cite{jgaa/OkamotoOUU11}.

    For the remainder of the proof assume that $c > 2k^3 + 4k$.
    Given a spanning tree $T$ of $G$ and an edge $e = \{u_1,u_2\} \in E(T)$,
    let $V_{u_1} \uplus V_{u_2} = V$ denote the subsets of vertices of $V$
    that belong to the connected component of $T - e$ containing vertices
    $u_1$ and $u_2$ respectively;
    notice that $V_{u_1}, V_{u_2}$ define a partition on $V$.
    The crux of our algorithm is the following claim,
    which allows us to upper-bound the value of $\stc(G)$.

    \begin{claim}\label{claim:distance_to_clique:upper_bound}
        There exists a spanning tree $T$ of $G$ with
        $\cng_G(T) < 2c - \frac{c}{k} + 2k^2$.
    \end{claim}

    \begin{claimproof}
        Let $S' = \setdef{s \in S}{|N_G(s) \cap C| > c - c/k} \subseteq S$ denote the set of vertices of $S$
        with more than $c - c/k$ neighbors in $C$.
        It holds that there exists $r \in C$ such that $S' \subseteq N_G(r)$.
        This follows by the pigeonhole principle, as the number of vertices in $C$ having at least one non-neighbor
        in $S'$ is at most $\sum_{s \in S'} |C \setminus N_G(s)| < \sum_{s \in S'}  c/k \le c$.
        Let $S_0 = N_G(r) \cap S$ denote the neighbors of $r$ in $S$, where $S_0 \supseteq S'$.

        We define $G' = (V,E')$ to be the subgraph of $G$ on edge set $E' = E(C \cup S_0, S \setminus S_0)$,
        and let $M \subseteq E'$ denote a maximum matching in $G'$,
        with $S_M \subseteq S \setminus S_0$ denoting the set of vertices of $S \setminus S_0$ that take part in $M$.
        Notice that for any $s \in S_M$, it holds that $|N_G(s) \cap C| \le c - \frac{c}{k}$,
        as $S_M \cap S' = \varnothing$.

        To construct the spanning tree $T$, we first include all edges incident to $r$ in $E(T)$
        and then add all edges present in $M$.
        Let $Z = S \setminus (S_0 \cup S_M)$.
        To conclude the construction of $T$,
        arbitrarily include some edges incident to the vertices of $Z$ in order to obtain a spanning tree.

        It remains to show that $\cng_{G,T}(e) < 2c - \frac{c}{k} + 2k^2$
        for all $e = \{u_1,u_2\} \in E(T)$.
        We consider different cases, depending on whether $r \in e$ or not.
        Before doing so, notice that for any such edge, it holds that
        \[
            \cng_{G,T}(e) \le \min\braces*{\sum_{u \in V_{u_1}} \deg_G(u) ,
            \sum_{u \in V_{u_2}} \deg_G(u)}.
        \]
        Additionally, we remark that $\forall z \in Z$, $|N_G(z) \cap C| \le k$, which implies that $\deg_G(z) < 2k$.
        To see this, notice that since $M$ is a maximum matching,
        no vertex in $Z$ is adjacent to a vertex of $C \cup S_0$ that does not take part in $M$.
        Thus, $|N_G(z) \cap C| \le |M| \le |S| \le k$.

        First consider the case where $r \notin e$.
        Then, it holds that either $C \subseteq V_{u_1}$ or $C \subseteq V_{u_2}$.
        Assume without loss of generality that $C \subseteq V_{u_1}$.
        In that case, notice that $V_{u_2} \cap S_0 = \varnothing$, while $|V_{u_2} \cap S_M| \le 1$
        and any other vertex of $V_{u_2}$ belongs to $Z$.
        Consequently, it follows that
        \[
            \cng_{G,T}(e) \le
            \sum_{u \in V_{u_2}} \deg_G(u)
            \le (c - \frac{c}{k} + k) + (k-1) \cdot 2k
            < c - \frac{c}{k} + 2k^2,
        \]
        with the first term accounting for the at most one vertex of $V_{u_2} \cap S_M$ and
        the second for the vertices of $V_{u_2} \cap Z$.

        Alternatively, $e = \{r, w\}$, where $w \in S_0 \cup (C \setminus \{r\})$.
        Notice that in this case, it holds that $|V_w \cap S_M| \le 1$,
        while any other vertex of $V_w$ belongs to $Z$.
        Consequently,
        \[
            \cng_{G,T}(e) \le
            \sum_{u \in V_{w}} \deg_G(u)
            < (c + k) + (c - \frac{c}{k} + k) + (k-2) \cdot 2k
            < 2c - \frac{c}{k} + 2k^2,
        \]
        with the first term accounting for $\deg_G(w)$, the second for the at most one vertex of $V_{w} \cap S_M$, and
        the third for the vertices of $V_{w} \cap Z$.
    \end{claimproof}

    The previous claim now allows us to determine the structure of a spanning tree $T_{\OPT}$ of $G$ of minimum congestion.

    \begin{claim}\label{claim:distance_to_clique:clique_vtcs_per_cc}
        For any $e = \{u_1,u_2\} \in E(T_{\OPT})$
        it holds that either $|V_{u_1} \cap C| \le 1$ or $|V_{u_2} \cap C| \le 1$.
    \end{claim}

    \begin{claimproof}
        Let $e = \{u_1,u_2\} \in E(T_{\OPT})$.
        By \cref{claim:distance_to_clique:upper_bound} and the fact that $c > 2k^3 + 4k$
        it follows that
        $\cng_{G, T_{\OPT}}(e) < 2c - \frac{c}{k} + 2k^2 < 2c - 2k^2 - 4 + 2k^2 = 2(c-2)$.
        Assume that $|V_{u_1} \cap C| = x$ and $|V_{u_2} \cap C| = c-x$, where $x \in [2, c-2]$.
        In that case,
        it holds that $\cng_{G, T_{\OPT}}(e) \ge x \cdot (c-x)$ due to the vertices of $C$.
        Since $x \cdot (c-x) \ge 2 \cdot (c-2)$ for all $x \in [2,c-2]$,
        it follows that $\cng_{G, T_{\OPT}}(e) \ge 2 (c-2)$, a contradiction.
    \end{claimproof}

    Consequently, due to \cref{claim:distance_to_clique:clique_vtcs_per_cc}
    it follows that there exists a vertex $r \in V$ such that in $T_{\OPT}-r$,
    every connected component contains at most one vertex of $C$.
    Notice that the vertices of $S$ are present in at most $k$ such connected components.
    Therefore, in order to compute $T_{\OPT}$, we initially guess vertex $r$;
    due to the indistinguishability of twin vertices, there are at most $2^k + k$ such choices -- $2^k$ in $C$ and $k$ in $S$.

    Next, we guess the set $Q \subseteq C$, which contains the vertices of $C$ present in connected components
    of $T_{\OPT}-r$ that contain vertices of $S$.
    Since $|Q| \le k$ and vertices belonging to the same twin classes are indistinguishable,
    there are at most $k \cdot 2^{k^2}$ choices regarding the set $Q$;
    $k$ choices for its cardinality, and then at most $(2^k)^k$ choices for the specific vertices.

    For the rest of the vertices $c \in C \setminus Q$, we add the edge $\{r,c\}$ in $E(T_{\OPT})$
    (if such an edge does not exist, our guess is wrong and we discard it).
    Finally, since $T_{\OPT}$ is a tree, it suffices to guess the parent for each vertex in $S \cup Q$,
    among the set $S \cup Q \cup \{r\}$, which gives a total of at most ${(2k+1)}^{2k}$ possibilities.
    Given one such spanning tree, one can in polynomial time compute its congestion.
    The final running time is $2^{\bO(k^2)} n^{\bO(1)}$.
\end{proof}

\subsection{Vertex Integrity}\label{ssec:vi}

Here we prove that the parameterization by vertex integrity renders {\STC} FPT.

\begin{theorem}\label{thm:vi}
{\STC} is fixed-parameter tractable parameterized by vertex integrity.
\end{theorem}

\begin{proof}
We say that two spanning trees $T$ and $T'$ of $G = (V,E)$ are \emph{equivalent}
if there is an automorphism $\eta$ of $G$ that is an isomorphism from $T$ to $T'$ as well.
Observe that the detour in $T$ for $\{u,v\} \in E(G)$ contains $\{x,y\} \in E(T)$
if and only if the detour in $T'$ for $\{\eta(u),\eta(v)\} \in E(G)$ contains $\{\eta(x), \eta(y)\} \in E(T')$.
This implies that equivalent spanning trees have the same congestion.

\begin{observation}\label{obs:equivalent-trees}
If $T$ and $T'$ are equivalent spanning trees of $G$,
then $\cng_{G}(T) = \cng_{G}(T')$.
\end{observation}

\proofsubparagraph*{Signatures of spanning trees.}
Let $k = \vi(G)$ and
$S \subseteq V(G)$ be a vertex set such that, for all components $C \in \cc(G-S)$,
it holds that $|S| + |V(C)| \le k$.
We say that $C, C' \in \cc(G-S)$ have the same \emph{type}
if there is an isomorphism $\eta$ from $C$ to $C'$
such that for every $v \in V(C)$, $N(v) \cap S = N(\eta(v)) \cap S$.
We call such an isomorphism a \emph{type isomorphism}.
Note that the relation of having the same type is an equivalence relation on $\cc(G-S)$.
Let $\mathcal{C}_{1}, \dots, \mathcal{C}_{p}$ be the equivalence classes defined by this relation.
That is, each $\mathcal{C}_{i}$ is a maximal subset of $\cc(G-S)$ with the same type.
For the number $p$ of different types it holds that $p \le 2^{k^{2}}$,
as there are at most $2^{k^{2}}$ non-isomorphic labeled graphs with at most $k$ vertices.

Let $C, C' \in \mathcal{C}_{i}$ for some $i \in [p]$.
Let $F$ and $F'$ be spanning forests of $C_{S} \coloneqq G[V(C) \cup S]$ and $C'_{S} \coloneqq G[V(C') \cup S]$, respectively.
We say that $F$ and $F'$ have the same \emph{forest type}
if there is an isomorphism $\eta$ from $C_{S}$ to $C'_{S}$ satisfying the following conditions:
\begin{itemize}
  \item $\eta(v) = v$ for every $v \in S$;
  \item $\eta$ is an isomorphism from $F$ to $F'$;
  \item $\eta$ restricted to $V(C)$ is a type isomorphism from $C$ to $C'$.
\end{itemize}
For $i \in [p]$,
let $q_{i}$ be the number of different forest types for $\mathcal{C}_{i}$.
This number can be upper-bounded by the number of possible ways of selecting at most $k-1$ edges from at most $\binom{k}{2}$ edges,
and thus $q_{i} \le k^{2k}$ holds.
We fix an ordering on the different forest types of such spanning forests in an arbitrary way.
This allows us to specify a forest type as ``the $j$-th forest type of $\mathcal{C}_{i}$'' for example.

Let $T$ be a spanning tree of $G$.
A component $C$ in $\cc(G - S)$ has \emph{type $(i,j)$} in $T$
if $C$ belongs to $\mathcal{C}_{i}$ and
$T[V(C) \cup S]$ has the $j$-th forest type of $\mathcal{C}_{i}$.
The \emph{signature} of $T$ is the pair $(T[S],\sigma_{T})$,
where $\sigma_{T}$ is a mapping such that
$\sigma_{T}(i,j)$ is the number of components $C \in \mathcal{C}_i$ of type $(i,j)$ in $T$.
We show that having the same signature implies the equivalence of two spanning trees.

\begin{claim}\label{lem:vi-tree-signature}
If spanning trees $T$ and $T'$ of $G$ have the same signature,
then $\cng_{G}(T) = \cng_{G}(T')$.
\end{claim}

\begin{claimproof}
Due to \cref{obs:equivalent-trees}, it suffices to show that $T$ and $T'$ are equivalent.
We show that there is an automorphism $\eta$ of $G$ that is an isomorphism from $T$ to $T'$.
Let $(F,\sigma)$ be the common signature of $T$ and $T'$.
For all $v \in S$, set $\eta(v) = v$.
Let $i \in [p]$ and $j \in [q_{i}]$.
We order the components of type $(i,j)$ in $T$ in an arbitrary way as $C_{1}, \dots, C_{\sigma(i,j)}$.
Similarly, we order the ones in $T'$ as $C'_{1}, \dots, C'_{\sigma(i,j)}$.
For each $h \in [\sigma(i,j)]$,
we extend $\eta$ by mapping $C_{h}$ to $C'_{h}$ with the type isomorphism certifying that they are of type $(i,j)$.
The definition of being type $(i,j)$ ensures that
the resulting mapping $\eta$ is an automorphism of $G$ and an isomorphism from $T$ to $T'$.
\end{claimproof}

\proofsubparagraph{Leaf components and non-leaf components in a spanning tree.}
Let $T$ be a spanning tree of $G$.
We call $C \in \cc(G-S)$ a \emph{non-leaf component} in $T$
if there exists $C' \in \cc(T[V(C)])$ such that $|N_{T}(V(C')) \cap S| \ge 2$;
otherwise, we call $C$ a \emph{leaf component} in $T$.
Intuitively, a non-leaf component connects some vertices in $S$, while a leaf component does not.
Observe that from the type $(i,j)$ in $T$ of a component, we can easily check if the component is a leaf or a non-leaf.

We can see that the number of non-leaf components is small as $T$ contains no cycle.

\begin{observation}\label{obs:vi-non-leaves}
There are at most $|S|-1$ non-leaf components in $T$.
\end{observation}

Let $H$ be the subgraph of $G$ induced by $S$ and the vertices in non-leaf components;
note that $H$ has at most $k^2$ vertices.
Observe that adding leaf components to $H$ never decreases the number of connected components in $T[V(H)]$.
Hence, the following holds.

\begin{observation}\label{obs:vi-local-spanning-tree}
$T[V(H)]$ is a spanning tree of $H$.
Furthermore, for each leaf component $C$ in $T$, $T[V(H) \cup V(C)]$ is a spanning tree of $G[V(H) \cup V(C)]$.
\end{observation}

By \cref{obs:vi-local-spanning-tree},
an edge in $H$ is detoured in $T[V(H)]$ and an edge in a leaf component $C$ is detoured in $T[V(H) \cup V(C)]$.
This implies the following locality of leaf components.

\begin{observation}\label{obs:vi-local-congestion-computation}
Let $C$ be a leaf component in $T$ and $e$ be an edge incident to a vertex in $C$.
If $e \in E(T)$, then every edge using $e$ in its detour in $T$ is incident to a vertex in $C$.
If $e \notin E(T)$, then the detour for $e$ in $T$ contains an edge $f$
if and only if the detour for $e$ in $T[V(H) \cup V(C)]$ contains $f$.
\end{observation}

By \cref{obs:vi-local-congestion-computation}, we can see that in any spanning tree $T$,
an edge incident to a vertex of a leaf component $C$ has congestion at most $|E(S \cup V(C))| < k^{2}$.
Therefore, if we know that $\stc(G) \ge k^{2}$, then we only have to bound the congestion of the edges in $T[V(H)]$.
This will be helpful for simplifying the description of the algorithm below.
In fact, we can assume that $\stc(G) \ge k^{2}$
since we can check whether $\stc(G) = s$ for every $s < k^{2}$
by using a fixed-parameter algorithm parameterized by $\tw(G)+\stc(G)$ (e.g., that of \cref{thm:tw-alg}).

\proofsubparagraph{The algorithm.}
We are ready to describe the algorithm, which can be outlined as follows.
\begin{enumerate}
  \item Guess the non-leaf components in an optimal spanning tree $T$.
  \begin{itemize}
    \item By \cref{obs:vi-non-leaves}, there are at most $|S|-1$ non-leaf components in $T$.
  \end{itemize}
  \item Let $H$ be the subgraph of $G$ induced by $S$ and the vertices in non-leaf components.
  Guess $T_{H} \coloneqq T[V(H)]$.
  \begin{itemize}
    \item By \cref{obs:vi-local-spanning-tree}, $T_{H}$ is a spanning tree of $H$.
  \end{itemize}
  \item Find the rest of the tree $T$.
  \begin{itemize}
    \item By \cref{lem:vi-tree-signature}, it suffices to find an optimal signature of $T$ consistent with the guesses.
    \item Since all remaining components are leaf components,
    \cref{obs:vi-local-congestion-computation} allows us to use integer linear programming (ILP) for this task.
  \end{itemize}
\end{enumerate}
As discussed above, we assume that $\stc(G) \ge k^{2}$,
and thus we only need to upper-bound the congestion of the edges in $T_{H}$.

For the non-leaf components, we have at most $p^{k} \le 2^{k^{3}}$ possibilities
for selecting at most $k-1$ components from $p$ classes.
For $i \in [p]$, let $\mathcal{C}'_{i}$ be the set of components obtained from $\mathcal{C}_{i}$
by removing the selected components (if no component is selected from $\mathcal{C}_{i}$, then $\mathcal{C}'_{i} = \mathcal{C}_{i}$).
Given the non-leaf components, there are at most $k^{2k^{2}}$ possibilities for $T_{H}$
as $H$ has at most $k^{2}$ vertices and $T_{H}$ is a spanning tree of $H$.
Thus, in total, there are at most $2^{k^{3}} \cdot k^{2k^{2}}$ candidates for $T_{H}$.
Therefore, by paying the multiplicative factor of $2^{k^{3}} \cdot k^{2k^{2}}$ in the running time,
we can assume that $T_{H}$ is correctly guessed.

Now we find the rest of the tree $T$ using ILP (the third step in the outline).
For each $i \in [p]$ and each $j \in [q_{i}]$,
let $x_{i,j}$ be the non-negative integer variable representing the value $\sigma_{T}(i,j)$;
that is, it determines the number of components of type $(i,j)$ in $T$.
If components of type $(i,j)$ are non-leaf in $T$, then we set $x_{i,j} = 0$.
For all $i \in [p]$, we set $\sum_{j \in [q_{i}]} x_{i,j} = |\mathcal{C}'_{i}|$
to make the number of used components in $\mathcal{C}'_{i}$ correct.

For each edge $e$ in $T_{H}$, take the positive integer variable $y_{e}$ that represents the congestion of $e$ in $T$.
Let $a_{e}$ be the number of edges in $H$ that contribute to the congestion of $e$ in $T$.
Let $b_{(i,j), e}$ be the number of edges in a component of type $(i,j)$ in $T$ that contribute to the congestion of $e$ in $T$.
By \cref{obs:vi-local-congestion-computation}, we can compute these constants in advance.
Now we can represent the congestion of $e$ in $T$ by setting $y_{e} = a_{e} + \sum_{i \in [p], \; j \in [q_{i}]} b_{(i,j), e} \cdot x_{i,j}$.

Finally, we set the objective of the formula to minimizing the maximum of $y_{e}$.
That is, we introduce a variable $z$ and put the constraint $y_{e} \le z$ for all $e \in E(H)$,
and then set $z$ as the objective function to be minimized.

The number of variables in the formula is at most $p \cdot \max_{i \in [p]} q_{i} + |E(H)| + 1$, which is $\bO(2^{k^{2}} k^{2{k}})$.
It is known that ILP parameterized by the number of variables is fixed-parameter tractable~\citep{mor/Lenstra83},
and thus we can conclude that our algorithmic task here is fixed-parameter tractable parameterized by $k$, the vertex integrity of $G$. This implies the theorem.
(Note that by a recent result of~\cite{focs/ReisR23},
ILP with $d$ variables can be solved in randomized ${(\log (2d))}^{\bO(d)}$ time.
Hence, this part can be done in randomized $k^{\bO(2^{k^{2}} k^{2{k}})}$ time.)
\end{proof}

\subsection{Feedback Edge Number}\label{ssec:fes}

Notice that the algorithm of \cref{thm:tw-alg} already implies that {\STC} is FPT
parameterized by the feedback edge number $\fes$ of the input graph:
any instance with $k \ge \fes+1$ is trivially yes, while if $k \le \fes$ we
can use the fact that $\tw(G) \le \fes(G)$ and apply \cref{thm:tw-alg}, obtaining
an $\fes^{\bO(\fes)} n^{\bO(1)}$ algorithm.
A natural question is whether the slightly super-exponential parametric dependence
can be overcome, and we answer this affirmatively by providing such an algorithm.
In fact, more strongly, we present a kernelization algorithm that results in a graph with only
$\bO(\fes)$ vertices and edges, thus allowing us to exhaustively guess the spanning tree of minimum congestion.
To do so, we notice that we can safely delete vertices of degree $1$,
resulting in a graph with only $\bO(\fes)$ vertices of degree larger than $2$.
Next, we can contract most of the remaining edges,
thereby leaving only $\bO(\fes)$ of them,
thus allowing us to guess exhaustively which edges belong to an optimal solution.

\begin{theorem}\label{thm:fes}
    {\STC} admits a kernel with $\bO(\fes)$ vertices and edges,
    where $\fes$ denotes the feedback edge number of the input graph.
\end{theorem}

\begin{proof}
    Let $G_0$ denote the input connected graph.
    We start by exhaustively deleting vertices of degree $1$ in $G_0$;
    it is easy to see that this is safe, as any such vertex is connected to
    any spanning tree via the single edge it is incident with,
    which is of congestion $1$.
    Let $G$ denote the resulting connected graph,
    where $\fes$ denotes its feedback edge number;
    notice that the deletion of vertices cannot increase
    the feedback edge number of a graph.
    Let $V^{=2} (G)$ and $V^{\ge 3} (G)$ denote the sets of vertices of $G$ of
    degree exactly $2$ and at least $3$ respectively;
    notice that they induce a partition on $V(G)$,
    since all vertices of $G$ are of degree at least~$2$.
    \cite[Lemma~2]{jgaa/BentertDKNN20} proved that
    $|V^{\ge 3} (G)| < 2\fes$ and
    $\sum_{v \in V^{\ge 3} (G)} \deg_G (v) =
    2(|V^{\ge 3} (G)| + \fes - 1) < 6 \fes$.%
    \footnote{We mention in passing a slight mistake in their proof,
    where they claim that $\sum_{v \in V^{\ge 3} (G)} \deg_G(v) = 3|V^{\ge 3} (G)|$
    instead.}
    We next define the following reduction rule.

    \proofsubparagraph*{Rule~$(\diamond)$.}
    Let $G$ be a graph with (not necessarily distinct) vertices $u,v \in V^{\ge 3}(G)$
    such that there is a $u$--$v$ path $P$ in $G$
    whose internal vertices all belong to $V^{=2}(G)$.
    Let $L_P \ge 1$ be equal to the number of internal vertices of $P$.
    Then,
    \begin{romanenumerate}
        \item\label{fes:red_rule} if $u=v$ and $L_P > 2$, contract edges in $P$ such that only $2$ internal vertices are left,
        \item if $u \neq v$, $\{u,v\} \in E(G)$, and $L_P > 1$, contract edges in $P$ such that only $1$ internal vertex is left,
        \item if $u \neq v$ and $\{u,v\} \notin E(G)$, delete all internal vertices of $P$ and add the edge $\{u,v\}$.
    \end{romanenumerate}

    Notice that Rule~$(\diamond)$ can be applied at most $n$ times, since each of its
    applications reduces the number of vertices of the graph.
    Let $G'$ denote the connected graph obtained after exhaustively doing so.
    Notice that $G$ can be obtained from $G'$ by only subdividing edges;
    edge subdivision does not change the spanning tree congestion
    of unweighted graphs~\citep[Lemma~7.10]{algorithmica/BodlaenderFGOL12},
    thus it holds that $\stc(G) = \stc(G')$.
    Furthermore, it is easy to see that $V^{\ge 3}(G') = V^{\ge 3}(G)$,
    and in fact any such vertex has the same degree in both $G$ and $G'$.
    Consequently, it follows that
    $\sum_{v \in V^{\ge 3} (G')} \deg_{G'}(v) = \sum_{v \in V^{\ge 3} (G)} \deg_{G}(v) < 6 \fes$.

    Let $E_1$ and $E_2$ define a partition on the edges of $G'$,
    where $E_1 = \setdef{e \in E(G')}{e \cap V^{\ge 3} (G') \neq \varnothing}$
    denotes the set of edges in $E(G')$ with at least one endpoint belonging to $V^{\ge 3} (G')$,
    and $E_2$ the rest, whose endpoints both belong to $V^{= 2} (G')$.
    Due to the previous discussion, it holds that $|E_1| < 6 \fes$.
    As for $E_2$, since Rule~$(\diamond)$ has been applied exhaustively,
    any such edge can only be due to Case~(\ref{fes:red_rule}) of Rule~$(\diamond)$,
    thus $|E_2| \le |E_1| / 2$ holds.
    In total, it follows that $|E(G')| < 9 \fes$.
    Furthermore, since $G'$ is connected, it follows that $|V(G')| \le |E(G')| + 1$.
\end{proof}

\section{Conclusion}\label{sec:conclusion}

In this paper we have presented a number of new results on the parameterized complexity of
\STC, painting an almost-complete picture regarding its tractability under the most standard parameterizations.
As a direction for future work, it would be interesting to consider the problem's tractability
parameterized by the neighborhood diversity,
the treewidth plus the maximum degree,
as well as whether an FPT-AS parameterized by clique-width exists.
Furthermore, although the problem is FPT by vertex cover due to \cref{thm:vi},
the algorithm is based on an ILP,
and a simpler (and faster) combinatorial algorithm might be possible under this
parameterization; along these lines, it would be interesting to determine whether the problem
admits a polynomial kernel in this case.
We additionally mention that it is unknown whether the problem remains NP-hard on cographs or
when $k=4$.

\bibliographystyle{abbrvnat}
\bibliography{refs}

\end{document}